\documentclass[sts]{imsart}

\usepackage{amsmath}
\usepackage{graphicx}     %!  setspace is used to control linepacing
\usepackage{amsfonts}
\usepackage{latexsym}
\usepackage{upgreek}
\usepackage{amssymb}
\usepackage{amsthm}
\usepackage{fancyhdr}
\usepackage{mathrsfs}
\usepackage{framed}
\usepackage{bm}
\usepackage{perpage} 
\usepackage{comment}
\usepackage{tikz}

\usetikzlibrary{positioning}
\usepackage{bbm}
\usepackage{enumerate}
\usepackage{dsfont}
\usepackage{color}

\numberwithin{equation}{section}

\newtheorem{theorem}{Theorem}[section]
\newtheorem{lemma}[theorem]{Lemma}
\newtheorem{example}[theorem]{Example}
\newtheorem{prop}[theorem]{Proposition}
\newtheorem{remark}[theorem]{Remark}

\newtheorem{assumption}[theorem]{Assumption}
\newtheorem{definition}[theorem]{Definition}

\newcommand{\propo}{\pi}
\newcommand{\tr}{{\rm {Tr}}}
\newcommand{\effd}{{\mathsf {efd}}}
\newcommand{\ESS}{{\mathsf {ess}}}
\newcommand{\nd}{d}
\newcommand{\Var}{\mbox{Var}}

\newcommand{\G}{N}
\newcommand{\R}{\mathbb{R}}
\newcommand{\rp}{\mathbb{P}}

\newcommand{\hiddensection}[1]{
\stepcounter{section}
\section*{}}

\newcommand{\n}{{N}}
\newcommand{\N}{\mathbb{N}}

\newcommand{\C}{C}

\newcommand{\h}{\mathcal{H}}

\newcommand{\F}{\mathcal{F}}

\newcommand{\coloneq}{\mathrel{\mathop:}=}

\newcommand{\norm}[1]{\left\|#1\right\|}

\newcommand{\sumj}{\sum_{j=1}^{\infty}}

\newcommand{\E}{\mathbb{E}}
\newcommand{\pr}[2]{\big\langle#1,#2\big\rangle}
\newcommand{\envelope}{(\raisebox{-.5pt}{\scalebox{1.45}{\Letter}}\kern-1.7pt)}

\newcommand{\sumi}{\sum_{j=1}^\infty}

\newcommand{\bigO}[1]{\mathcal{O}{\left(#1\right)}}

\newcommand{\const}{\propo(g)}

\newcommand{\moda}[1]{\textcolor{brown}{#1}\index {#1}} %ANDREW
\definecolor{brown}{rgb}{0.6,0.5,0.0}
\definecolor{darkgreen}{rgb}{0,.6,0}
\definecolor{darkyellow}{rgb}{1,.5,.3}
\newcommand{\si}[1]{\lambda}

\newcommand{\sumn}{\sum_{n=1}^{N}}

\newcommand{\is}{\tm^N} %particle approximation measure
\newcommand{\mc}{\propo^N_{\mbox{\tiny{\rm MC}}}} 

\newcommand{\Zn}{\mc(g)}
\newcommand{\Zpi}{\propo(g)}
\newcommand{\test}{\phi}
\newcommand{\testc}{\overline{\phi}}

\newcommand{\wn}{w^n} % normalized weight
\newcommand{\wN}{w^{(N)}} % maximal normalized weight
\newcommand{\un}{u^n} % particle
\newcommand{\um}{u^m} % particle in sums
\newcommand{\ui}{u(i)} % the i'th dim of a particle

\newcommand{\X}{\mathcal{X}}
\newcommand{\Y}{\mathcal{Y}}

\newcommand{\tm}{\mu}
\newcommand{\prm}{\pi}
\newcommand{\rn}{g}

\newcommand{\prior}{\rp_u}
\newcommand{\post}{\rp_{u|y}}
\newcommand{\noise}{\rp_{\eta}}
\newcommand{\like}{\rp_{y|u}}

\newcommand{\prc}{\Sigma}

\newcommand{\posc}{C}
\newcommand{\nc}{\Gamma}
\newcommand{\ns}{\gamma}
\newcommand{\fp}{K}
\newcommand{\posm}{m}
\newcommand{\rnip}{g} %RADON-NIKODYM FOR IP's (same for now)

\startlocaldefs
\endlocaldefs

\newcommand{\exm}{m} %mean in introduction to Gaussian measures
\newcommand{\auxm}{S} 
 
\newcommand{\excov}{\mathcal{C}}
\newcommand{\excoef}{\kappa}

\newcommand{\signalmap}{M}
\newcommand{\zeroc}{P}
\newcommand{\signalc}{Q}
\newcommand{\obsop}{H}
\newcommand{\obsc}{R}

\newcommand{\dkl}{D_{\mbox {\tiny{\rm KL}}}}

\newcommand{\law}{\rp}
\newcommand{\mar}{\nu_{y}}

\usepackage{xr}
\begin{document}

\begin{frontmatter}

% "Title of the paper"
\title{Importance Sampling: Intrinsic Dimension and Computational Cost}
\runtitle{Importance Sampling} 

% indicate corresponding author with \corref{}
% \author{\fnms{John} \snm{Smith}\corref{}\ead[label=e1]{smith@foo.com}\thanksref{t1}}
% \thankstext{t1}{Thanks to somebody} 
% \address{line 1\\ line 2\\ printead{e1}}
% \affiliation{Some University}

\begin{aug}
\author{\fnms{S.} \snm{Agapiou}\thanksref{t1}\ead[label=e1]{first@somewhere.com}},
\author{\fnms{O.} \snm{Papaspiliopoulos}\ead[label=e2]{second@somewhere.com}}\thanksref{t2}
\author{\fnms{D.} \snm{Sanz-Alonso}\ead[label=e2]{second@somewhere.com}}\thanksref{t3}
\and
\author{\fnms{A. M.} \snm{Stuart}\thanksref{t4}
\ead[label=e3]{third@somewhere.com}
\ead[label=u1,url]{www.foo.com}}
\thankstext{t1}{ Department of Mathematics and Statistics, University of Cyprus, 1 University Avenue, 2109 Nicosia, Cyprus.}
\thankstext{t2}{ICREA \& Department of Economics, Universitat Pompeu Fabra, Ramon Trias Fargas 25-27, 08005 Barcelona, Spain.}
\thankstext{t3}{Division of Applied Mathematics, Brown University, Providence RI02906, USA.}
\thankstext{t4}{Computing and Mathematical Sciences, California Institute of
Technology, Pasadena CA91125, USA.}
%\runauthor{F. Author et al.}

\address{agapiou.sergios@ucy.ac.cy, omiros.papaspiliopoulos@upf.edu, \text{daniel\_ \hspace{-3pt}sanz-alonso1@brown.edu},  astuart@caltech.edu}
\end{aug}

\runauthor{Agapiou, Papaspiliopoulos, Sanz-Alonso, Stuart}

%\begin{abstract}
%\end{abstract}

%\begin{keyword}[class=MSC]
%\kwd[Primary ]{}
%\kwd{}
%\kwd[; secondary ]{} 
%\end{keyword}

%\begin{keyword}
%\kwd{}
%\kwd{}
%\end{keyword}

\begin{abstract}
\begin{comment}
The basic idea of importance sampling is to use independent samples from
one measure in order to approximate expectations with respect to another
measure. Understanding how many samples are needed is key to understanding
the computational cost of the method, and hence to understanding
when it will be effective and when it will not. It is intuitive that 
the size of the difference between the measure which is sampled, 
and the measure against which expectations are to be computed, is key to the
computational cost. An implicit challenge in many of the published
works in this area is to find useful quantities which measure this difference
in terms of parameters which are pertinent for the practitioner. The subject
has attracted substantial interest recently from within a variety of
communities. The objective of this paper is to overview and unify the
resulting literature in the area by creating an overarching framework.  
The general setting is studied in some detail, followed by deeper development
in the context of Bayesian inverse problems and filtering.
\end{comment}

The basic idea of importance sampling is to use independent samples from
a proposal measure in order to approximate expectations with respect to a target
measure. It is key to understand how many samples are required in order to guarantee accurate approximations.
Intuitively, some notion of distance between the target and the proposal should determine the
computational cost of the method. A major challenge is to quantify this distance
in terms of parameters or statistics that are pertinent for the practitioner. The subject
has attracted substantial interest from within a variety of
communities. The objective of this paper is to overview and unify the
resulting literature by creating an overarching framework.  
A general theory is presented, with a focus on the use of importance sampling in
Bayesian inverse problems and filtering.
\end{abstract}

\end{frontmatter}

\section{Introduction}
\label{sec:INT}

\subsection{Our Purpose}
\label{ssec:pur}

Our purpose in this paper is to overview various ways of measuring
the computational cost of importance sampling, to link them 
to one another through transparent mathematical reasoning, and to create
  cohesion in  the vast published literature on this subject. 
In addressing these issues we will study importance sampling
in a general abstract setting, and then in the particular cases
of Bayesian inversion and filtering. These two application
settings are particularly important as there are many pressing scientific,
technological and societal problems which can be formulated via
inversion or filtering. An example of such an inverse problem is the determination
of subsurface properties of the Earth from surface measurements;
an example of a filtering problem is the assimilation of atmospheric
measurements into numerical weather forecasts. We now proceed to overview the subject of importance sampling, and
the perspective on it that is our focus. In subsection \ref{ssec:org} we describe
the organization of the paper and our main contributions. 
Subsection \ref{ssec:rev} then collects all the references linked to the material in 
the introduction, as well as other general references on importance 
sampling. Each subsequent section of the paper contains its own literature 
review subsection providing further elaboration of the literature,
and linking it to the details of the material that we present in that
section.

The general abstract setting in which we work is as follows. 
%$(\X, {\mathcal F})$ be a measurable space and 
We let $\tm$ and $\prm$  be two probability measures on a measurable space $(\X, {\mathcal F})$
related via the expression
\begin{equation}\label{abscont}
\frac{d\tm}{d\prm}(u):=g(u)\Big/\int_{\X} g(u)\prm(du).\end{equation}
Here, $g$ is the unnormalised  {\em density (or Radon-Nikodym derivative)} 
of $\tm$ with respect to $\prm.$ 
Note that the very existence of the density implies that the
target is {\em absolutely continuous} with respect to the proposal;
absolute continuity will play an important role in our subsequent
developments. 

Importance sampling is a method for using independent samples from
the {\em proposal} $\prm$ to approximately compute expectations with
respect to the {\em target} $\tm.$ The way importance sampling
  (and more generally Monte Carlo integration methods) is used within
  Bayesian statistics and Bayesian inverse problems is as an
  approximation of the target measure $\tm$ by a random probability
  measure using weighted samples that are generated from $\prm$.  (This perspective differs from that arising in other disciplines, e.g. in certain applications in mathematical finance, such as option pricing).  Our perspective is dictated by the need to use the samples to estimate expectations and quantiles of a wide range of functions defined on the state space, e.g. functions of a single variable or pairs of variables, or marginal likelihood quantities. The resulting approximation is typically called a particle approximation. Our perspective on importance sampling as a probability measure approximation dictates in turn the tools for studying its performance. Its computational cost is measured by the number of samples required to control the worst error made when approximating expectations within a class of test functions. In this article, and following existing foundational work,  we primarily focus on a total variation metric between random measures
for assessing the particle
approximation error.  Intuitively, the size of the error  is related to how far
the target measure is from the proposal measure. We make this
intuition precise, and connect the particle approximation error to  a
key quantity, the second moment of $d \tm/d \prm$ under the proposal,
which we denote by $\rho$: 
\[
\rho = \pi(g^2) / \pi(g)^2\,.
\] 
As detailed below, $\rho$ is essentially the $\chi^2$ divergence between the
target and the proposal.

The first application of this setting that we study
is the linear inverse problem to determine $u \in \X$ from $y$ where 
\begin{equation}\label{inverseproblem}
y = \fp u + \eta, \quad \eta\sim \G(0,\nc).
\end{equation}
We adopt a Bayesian approach in which we place a prior 
$u\sim \prior= \G(0,\prc)$, assume that $\eta$ is independent
of $u$, and seek the posterior $u|y \sim \post.$ We study importance sampling
with $\post$ being the target $\tm$ and $\prior$ being
the proposal $\prm$.

The second application is the linear filtering problem of sequentially
updating the distribution of $v_j \in \X$ given $\{y_{i}\}_{i=1}^{j}$ where 
\begin{align}\label{Filteringproblem}
\begin{split}
v_{j+1} &= M v_j + \xi_j, \quad \xi_j \sim \G(0,Q),\quad j\ge0,\\
y_{j+1} &= H v_{j+1} + \zeta_{j+1}, \quad \zeta_{j+1} \sim  \G(0,R), \quad j\ge 0.
\end{split}
\end{align} 
We assume that the problem has a Markov structure.
We study the approximation of one step of the filtering update by means
of particles, building on the study of importance
sampling for the linear inverse problem. 
To this end it is expedient to work on the product space $\X \times \X$,
and consider importance sampling for $(v_j,v_{j+1}) \in \X \times \X$. 
It then transpires that, for two different proposals, which are commonly termed
the {\em standard proposal} and the  {\em optimal proposal},
the cost of one step of particle filtering 
may be understood by the study of a linear inverse problem on $\X$;
we show this for both proposals, and then use the link to an
inverse problem to derive results about the cost of particle filters based on these
 two proposals. 

The linear Gaussian models that we study can —and typically should— be
treated by direct analytic calculations or efficient simulation of Gaussians.
However, it is possible to analytically study the dependence of $\rho$
on key parameters within these model classes, and furthermore they are
flexible enough to incorporate formulations on function spaces, and  
their finite dimensional approximations. Thus, they are an excellent
framework for obtaining insight into the performance of importance sampling
for inverse problems and filtering.

For the abstract importance sampling problem we will
relate $\rho$
to a number of other natural quantities. These include the {\em effective
sample size} $\ESS$, used heuristically in many application domains,
and a variety of {\em distance metrics} between $\prm$ and $\tm$. 
Since the existence of a density between target and proposal plays an important role
in this discussion, we will also investigate what happens as this absolute
continuity property breaks down. We study this first in 
 {\em high dimensional problems}, and second in
{\em singular parameter limits} (by which we mean limits of
important parameters defining the problem).
The ideas behind these two different ways of breaking absolute continuity
 are presented in the general framework, and then substantially 
developed in the inverse problem and filtering settings. The motivation for studying these limits
can be appreciated by considering 
%The motivation for studying high dimensional problems can
%be appreciated by  considering 
the two examples mentioned at
the start of this introduction: inverse problems from the Earth's
subsurface, and filtering for numerical weather prediction. In
both cases the unknown which we are trying to determine from data
is best thought of as a spatially varying field for subsurface
properties such as permeability, or atmospheric properties, such
as temperature. In practice the field will be discretized
and represented as a high dimensional vector, for computational
purposes, but for these types of application the state dimension
can be of order $10^9$. Furthermore as computer power advances
there is pressure to resolve more physics, and hence for
the state dimension to increase. Thus, it is important to understand 
infinite dimensional problems, and sequences of approximating
finite dimensional problems which approach the infinite dimensional
limit. 
A motivation for studying singular parameter
limits arises, for example, from problems in which the noise
is small and the relevant log-likelihoods scale inversely with the noise
variance.
\begin{comment}
Breakdown of absolute continuity will be related to limits in 
which the target and proposal become increasingly close to
being {\em mutually singular.}
\end{comment}

\begin{comment}
We will highlight a variety of notions of {\em intrinsic dimension}; these 
may differ substantially from the dimensions of the spaces where 
the unknown $u$ and the data $y$ live. 
We then go on to show how these intrinsic dimensions relate to the 
parameter $\rho$, previously demonstrated to be
central to computational cost.  We do so in various limits arising from 
large dimension of $u$ and $y$, and/or small observational
noise. 
We also link these concepts to breakdown of absolute continuity.
Finally we apply our understanding of linear inverse problems
to particle filters, translating the results from one to the other
via the correspondence between the two problems, for both standard
and optimal proposals, as described above.
\end{comment}

\begin{comment}
The main idea underpinning our presentation of the cost of
importance sampling is that the size of the second moment of the 
Radon-Nikodym derivative between proposal and target is the key quantity.
This leads naturally to a study of the breakdown of absolute continuity
between proposal and target in various limits. These include large dimensional limits
and  singular parameter 
limits relating, for example, to scaling of the observational noise.
\end{comment}

This paper aims in particular to contribute towards a better understanding of the recurrent claim that importance sampling 
suffers from the curse of dimensionality. Whilst there is some empirical 
truth in this statement, there is a great deal of confusion in the literature 
about what exactly makes importance sampling hard. In fact such a statement 
about the role of dimension is vacuous unless ``dimension'' is defined 
precisely. We will substantially clarify these issues in the contexts of inverse problems and filtering. Throughout this paper we use the following conventions:
\begin{itemize}
\item State space dimension is the dimension of the measurable space where the measures $\tm$ and $\propo$ are defined. We will be mostly interested in the case where the measurable space $\X$ is a separable Hilbert space, in which case the state space dimension is the cardinality of an orthonormal basis of the space. In the context of inverse problems and filtering, the state space dimension is the dimension of the unknown.
\item Data space dimension is the dimension of the space where the data lives.
\item Nominal dimension is the minimum of the state space dimension and the data state dimension.
\item Intrinsic dimension: we will use two notions of intrinsic dimension for linear Gaussian inverse problems, denoted by  $\effd$ and $\tau$. These combine state/data dimension and small noise parameters. They can be interpreted as a measure of how informative the data is relative to the prior.
\end{itemize}

We show that the intrinsic dimensions are
natural when studying the computational cost of importance
sampling for inverse problems. In particular we show how these intrinsic dimensions relate to the 
parameter $\rho$ introduced above, a parameter that we show to be 
central to the computational cost, and to the breakdown of absolute continuity. 
Finally we apply our understanding of linear inverse problems
to particle filters, translating the results from one to the other
via an interesting correspondence between the two problems, for both standard
and optimal proposals, that we describe here.
In studying these quantities, and their inter-relations,
we aim to achieve the purpose set out at the start of this introduction.
\begin{comment}
Furthermore, a bibliography subsection, within each section,
will link our overarching mathematical framework to the published
literature in this area.
\end{comment}

\subsection{Organization of the Paper and Main Contributions} 
\label{ssec:org}

Section \ref{sec:IS} describes importance sampling in abstract form.
In sections \ref{sec:BIP} and \ref{sec:FIL} the linear Gaussian
inverse problem and the linear Gaussian filtering problem are
studied. Our aim is to provide a digestible narrative
and hence all proofs ---and all technical matters related to studying measures in infinite dimensional spaces--- are left to the Supplementary Material. 

Further to providing a unified narrative of the existing literature, this paper contains some original contributions that shed new light on the use of importance sampling for inverse problems and filtering. Our main new results are:

\begin{itemize}
\item Theorem \ref{non-asymptotictheorem} bounds the error of
  importance sampling for bounded test functions. The main appeal of
  this theorem is its non-asymptotic nature, together with its clean
  interpretation in terms of: (i) the key quantity $\rho$; (ii)
  effective sample size; (iii) metrics between probability measures;
  (iv) existing asymptotic results.  According to  the perspective on
    importance sampling as an approximation of one probability measure
    by another, the metric used in Theorem \ref{non-asymptotictheorem} is natural and it has already been
    used in important theoretical developments in the field as we
    discuss in section \ref{ssec:lit}. On the other hand, the result
    is less useful for quantifying the error for a specific test
    function of interest, such as linear, bilinear or quadratic
    functions, typically used for computing moments and
    covariances. We discuss extensions and generalizations  in section \ref{sec:IS}.
%The drawback is the limitation to bounded test functions; this issue is addressed in Theorem \ref{thm:nonasmom} which
%also serves to illustrate the fact that a universal analysis of importance
%sampling in terms of $\rho$ alone is not possible.
\item Theorem \ref{maintheorem} studies importance sampling for inverse problems. It is  formulated in the linear 
Gaussian setting to allow a clear and full development of the 
connections that it makes between heretofore disparate notions. In particular we highlight the following. (i) It provides the first clear connection between finite intrinsic dimension and absolute continuity between posterior and prior. (ii) It demonstrates the relevance of the intrinsic dimension 
--rather than the state space or the nominal dimension-- in the performance of 
importance sampling, by linking the intrinsic dimension and the 
parameter $\rho$; thus it shows the combined effect of the prior, 
the forward map, and the noise model in the efficacy of the method. 
(iii) It provides theoretical support for the use of algorithms based 
on importance sampling for posterior inference in function space, 
provided that the intrinsic dimension is finite and the value 
of $\rho$ is moderate.
\item Theorems \ref{definitionrelatedinverseproblem} and \ref{t:2} are proved by studying the inverse problem at the heart of importance
sampling based particle filters. These theorems, together with Theorem \ref{theoremcomparisoncollapse} and Example \ref{EX}, provide an improved understanding of the advantages 
of the optimal proposal over the standard proposal in the context of filtering.
\end{itemize}

\subsection{Literature Review}
\label{ssec:rev}
In this subsection we provide a historical review of the 
literature in importance sampling. Each of the following
sections \ref{sec:IS}, \ref{sec:BIP}, and \ref{sec:FIL} will contain a further literature
review subsection providing detailed references linked
explicitly to the theory as outlined in those sections.

Early developments of importance sampling as a method to reduce the variance in Monte Carlo estimation date back to the early 1950's \cite{kahn1953methods}, \cite{kahn1955use}. In particular the paper 
\cite{kahn1953methods} demonstrates how to optimally choose the proposal density for given test function $\phi$ and target density. Standard text book references for importance sampling include \cite{doucet2001introduction} and \cite{liu2008monte}. Important methodological improvements were introduced in \cite{li2013two}, \cite{owen2000safe}, \cite{liu1998sequential}, and \cite{tan2004likelihood}. A modern view of importance sampling in the general framework \eqref{abscont} is given in \cite{omirosb}. A comprehensive description of Bayesian inverse problems in finite state/data space dimensions can be found in \cite{KS05}, and its formulation in infinite dimensional spaces in \cite{DS15,SL07,SL12a,SL12b,AS10}. Text books
overviewing the subject of filtering and particle filters include \cite{del2004feynman,BC09},
and the article \cite{CD02} provides a readable introduction to the
area. For an up-to-date 
and in-depth  survey of nonlinear filtering see
\cite{crisan2011oxford}. The linear Gaussian inverse problem and the
linear Gaussian filtering problem have been extensively studied  because they arise naturally
in many applications, lead to considerable algorithmic
tractability, and  provide theoretical insight. 
For references concerning linear Gaussian inverse problems
see \cite{JF70,AM84,LPS89,KLS15}. 
The linear Gaussian filter --the Kalman filter-- was introduced
in \cite{kalman1960new}; see \cite{lancaster1995algebraic} for
further analysis.  The inverse problem of determining subsurface properties of the Earth from surface measurements is discussed in \cite{ORL08}, while the filtering problem of assimilating atmospheric measurements for numerical weather prediction is discussed in \cite{K03}.

The key role of $\rho$, the second moment of the Radon-Nikodym derivative between the target and the proposal,  has long been acknowledged \cite{liu1996metropolized},  \cite{pitt1999filtering}. The crucial question of how to choose a proposal measure that leads to small value of $\rho$ has been widely studied, and we refer to \cite{liang2007stochastic} and references therein.    In this vein,  our theory in sections \ref{sec:BIP} and \ref{sec:FIL} shows precise conditions that guarantee $\rho<\infty$ in inverse problems and filtering settings, in terms of well-defined basic concepts such as absolute continuity of 
the target with respect to the proposal. Our study of importance sampling for inverse problems in section \ref{sec:BIP} is limited to the choice of prior as proposal, which is of central theoretical relevance.  In practice, however, 
more sophisticated proposals are often used, potentially leading to reduced parameter $\rho$;  two novel ideas include the implicit
sampling method described in \cite{tu2013implicit}, and the use of proposals based on the ensemble Kalman filter suggested in \cite{le2009large}.
 The value of $\rho$  is known to be asymptotically linked to the effective sample size \cite{kong1992note}, \cite{kong1994sequential}, \cite{liu1996metropolized}. Recent justification for the use of the effective sample size within particle filters is given in \cite{whiteley2016role}. We provide a further non-asymptotic justification of the relevance of $\rho$
through its appearance in error bounds on the error in importance sampling;
a relevant related paper is \cite{CDL98}
which proved non-asymptotic bounds on the error in the importance-sampling
based particle filter algorithm. In this paper we will also bound the 
importance sampling error in terms of different notions of distance 
between the target and the proposal measures. Our theory is based on the $\chi^2$ divergence ---as in \cite{chen2005another}--- while the recent complementary analysis of importance sampling in \cite{CP15} highlights the advantages of the Kullback-Leibler divergence; a useful overview of the subject of distances between probability measures is \cite{gibbs2002choosing}.

We formulate problems in both finite dimensional and infinite dimensional state
spaces. We refer to \cite{OK02} for a modern presentation of probability
appropriate for understanding the material in this article. Some of our 
results are built on the rich area of Gaussian measures in Hilbert space; 
we include all the required background in the Supplementary Material, and references are included there.
However we emphasize that the presentation in the
main body of the text is designed to keep technical material to a 
minimum and to be accessible to readers who are not versed in the
theory of probability in infinite dimensional spaces.
Absolute continuity of the target with respect to the proposal ---or
the existence of a density of the target with respect to the
proposal--- is central to our developments. This concept
also plays a pivotal role in the understanding of 
Markov chain Monte Carlo (MCMC) methods in high 
and infinite dimensional spaces \cite{Tie}. A key
idea in MCMC is that  breakdown of absolute continuity on sequences
of problems of increasing state space dimension is responsible for poor
algorithmic performance with respect to increasing dimension; this should
be avoided if possible, such as for problems with a well-defined
infinite dimensional limit  \cite{cotter2013mcmc}. Similar
ideas will come into play in this paper.

As well as the breakdown of absolute continuity through increase
in dimension, small noise limits  can also lead to sequences of
proposal/target measures which
are increasingly close to mutually singular and for which
absolute continuity breaks down.
Small noise regimes are of theoretical and computational interest for both inverse problems and filtering. For instance, in inverse problems there is a growing interest in the study of the concentration rate of the posterior in the small observational noise limit, see \cite{KVZ12}, \cite{ALS13}, \cite{KVZ13}, \cite{ASZ14}, \cite{KR13}, \cite{SV13}, \cite{KLS15}.
In filtering and multiscale diffusions, the analysis and development of improved proposals in small noise limits is an active research area  \cite{vanden2012data}, \cite{zhang2013importance}, \cite{dupuis2012importance}, \cite{spiliopoulos2013large} \cite{tu2013implicit}.

In order to quantify the computational cost of a problem, a recurrent concept is that of intrinsic dimension. Several notions of intrinsic
dimension have been used in different fields, including dimension of 
learning problems  \cite{bishop}, \cite{zhang2002effective}, \cite{zhang2005learning}, of statistical inverse problems \cite{lu2014discrepancy}, of functions in the context of quasi Monte Carlo (QMC) integration in finance applications \cite{C97}, \cite{MC96}, \cite{kuo2005lifting}, and of data assimilation problems \cite{chorin2013conditions}. The underlying theme is that in many application areas where models are formulated in high dimensional state spaces, there is often a small subspace which captures most of the features of the system. It is the dimension of this subspace that effects the cost of the problem. The recent  subject of \emph{active subspaces} shows promise in finding such low dimensional subspace of interest in certain applications \cite{constantine2015active}.  
 In the context of inverse problems, the paper \cite{BBL08} proposed a 
notion of intrinsic dimension that was shown to have a direct connection with the performance of importance sampling. We introduce a further notion of intrinsic dimension for Bayesian inverse problems which agrees with the notion of effective number of parameters used in machine learning and statistics \cite{bishop}. We also establish that this notion of dimension and the one in \cite{BBL08} are finite, or otherwise, at the same time. Both intrinsic 
dimensions account for three key features of the cost of the
inverse problem: the nominal dimension (i.e. the minimum of the
dimension of the state space and the data), the size of the
observational noise and the regularity of the prior relative to the
observation noise. Varying the parameters related to 
these three features may cause a breakdown of absolute continuity. The deterioration of importance sampling in large nominal dimensional limits has been widely investigated \cite{BBL08}, \cite{BLB08}, \cite{snyder2008obstacles}, \cite{snyderberngtsson}, \cite{snyder2011alone}, \cite{slivinskipractical}. In particular, the key role of the intrinsic dimension, rather than the nominal one, in explaining this deterioration was studied in \cite{BBL08}.
 Here we study the different behaviour of importance sampling as absolute continuity is broken in the three regimes above, and we investigate 
whether, in all these regimes, the deterioration of importance sampling 
may be quantified by the various intrinsic dimensions that we introduce.

We emphasize that, whilst the theory and discussion in section \ref{sec:IS}
is quite general, the applications to Bayesian inverse problems (section
\ref{sec:BIP}) and filtering (section \ref{sec:FIL})
are in the case of linear problems with additive Gaussian noise. This
linear Gaussian setting allows substantial explicit calculations and yields
considerable insight. However empirical evidence related to the behaviour
of filters and Monte Carlo based methods  when applied to nonlinear
problems and non-Gaussian target measures suggests
that similar ideas may apply in those situations;
see \cite{snyder2008obstacles,cotter2013mcmc,bui2012analysis,cui2014likelihood,constantine2015active}. Quantifying this empirical experience
more precisely
is an interesting and challenging direction for future study. We note
in particular that extensions of the intrinsic dimension quantity that
we employ have been provided in the literature for Bayesian
hierarchical non-Gaussian models, more specifically within the so-called
deviance information criterion of \cite{dic}, see Section
\ref{ssec:stat-dim} for more discussion.

\subsection{Notation}\label{ssec:notation}
Given a probability measure $\nu$ on a  measurable space $(\X, {\mathcal F})$  expectations of a measurable function $\test: \X\to \R$ with respect to $\nu$ will be written as both $\nu(\test)$ and $\E_\nu [\test].$ When it is clear which measure is being used we may drop the suffix $\nu$ and write simply $\E[\test].$ Similarly, the variance will be written as $\Var_\nu(\test)$ and again we may drop the suffix when no confusion arises from doing so. All test functions $\phi$ appearing in the paper are assumed to be measurable.

We will be interested in sequences of measures indexed by time or by the state space dimension. These are denoted with a subscript, e.g. $\nu_t$, $\nu_d.$ Anything to do with samples from a measure is denoted with a superscript: $N$ for the number of samples, and $n$ for the 
indices of the samples. The $i$-th coordinate of a vector $u$ is denoted by $u(i).$ Thus, $u_t^n(i)$ denotes $i$-th coordinate of the $n$-th sample from the measure of interest at time $t.$ 
Finally, the law of a random variable $v$ will be denoted by $\law_v.$

\section{Importance Sampling}
\label{sec:IS}

In subsection \ref{ssec:gen} we define importance sampling and in subsection \ref{ssec:mom} we demonstrate
the role of the second moment of the target-proposal density, $\rho$; we
prove two non-asymptotic theorems showing 
${\mathcal O}\bigl((\rho/N)^{\frac12}\bigr)$ 
convergence rate of importance sampling
with respect to the number $N$ of particles. Then in subsection \ref{ssec:ess} we show how $\rho$
relates to the effective sample size $\ESS$ as often defined by practitioners,
whilst in subsection \ref{ssec:met} we link
$\rho$ to various distances  between probability measures.
In subsection \ref{ssec:con} we highlight the role of the breakdown of
absolute continuity in the growth of $\rho$, as the dimension of the
space $\X$ grows. Subsection \ref{ssec:sin} follows with a similar
discussion relating to singular limits of the density between
target and proposal. Subsection \ref{ssec:lit} contains a literature 
review and, in particular, sources for all the material in this section.

\subsection{General Setting}
\label{ssec:gen}

We consider target $\tm$ and proposal $\propo$, both probability measures on 
the  measurable space $(\X, {\mathcal F})$, related by \eqref{abscont}. 
In many statistical applications interest lies in estimating expectations 
under $\tm$, for a collection of test functions, using samples from $\prm$. 
For a test function $\test:\X\to\R$ such that $\tm(|\test|) < \infty$, 
the identity 
\begin{equation*}
\tm(\test)= \frac{\propo(\test \rn)}{\const} ,
\end{equation*}
leads to the {\em autonormalized importance sampling} estimator:
\begin{align}\label{autonormalizedestimator}
\is(\test) &: =\frac{\frac1N \sum_{n=1}^N \test(\un)\rn(\un)}{\frac1N \sum_{m=1}^N \rn(\um)},\quad \un\sim\propo\quad \text{i.i.d.} \\
&=\sum_{n=1}^N\wn  \test(\un), \quad \quad \wn:=\frac{\rn(\un)}{\sum_{m=1}^N
  \rn(\um)};  \nonumber
\end{align} 
here the $\wn$'s are called the \emph{normalized weights}.  As suggested by the
notation, it is useful to view \eqref{autonormalizedestimator} as
integrating a function $\test$ with respect to the random probability
measure $\is := \sum_{n=1}^N \wn \delta_{\un}.$ 
Under this perspective, importance sampling consists of approximating
the target $\tm$ by the measure $\is$, which is typically called the \emph{particle
approximation of $\tm$}. Note that, while $\is$ depends on
the proposal $\propo$,  we suppress this
dependence for economy of notation.  
Our aim is to understand the quality of the approximation $\is$ of $\tm$.
In particular we would like to know how large to choose $N$ in order
to obtain small error. This will quantify the computational cost
of importance sampling.

\subsection{A Non-asymptotic Bound on Particle Approximation Error} 
\label{ssec:mom}

A fundamental quantity in addressing this issue is $\rho$, defined by
\begin{equation}\label{rhodef}
\rho:=\frac{\propo(\rn^2)}{\const^2}\,.
\end{equation}
Thus $\rho$ is the second moment of the Radon-Nikodym derivative 
of the target with respect to the proposal. 
The Cauchy-Schwarz inequality shows that $\pi(g)^2 \le \pi(g^2)$ and 
hence that $\rho \ge 1.$ 
Our first non-asymptotic result shows that, for bounded test functions
$\test$, both the bias and the mean square error (MSE) of the autonormalized importance
sampling estimator are $\bigO {N^{-1}}$ with constant of proportionality
linear in $\rho.$ 
 
\begin{theorem}\label{non-asymptotictheorem}
Assume that $\tm$ is absolutely
continuous with respect to $\propo$, with square-integrable density
$\rn$, that is, $\propo(\rn^2) < \infty$.
The bias and MSE of importance sampling over bounded
test functions may be characterized as follows:
\begin{equation*}
\sup_{|\test|\le 1}\Bigl|\E \bigl[\is(\test) - \tm(\test)\bigr]\Bigr| \le
\frac{12}{N}\rho \,, 
\end{equation*}
and
\begin{equation*}
\sup_{|\test|\le 1}\E \left[\left(\is(\test) - \tm(\test)\right)^2\right]
\le \frac4N\rho. 
\end{equation*}
\end{theorem}

\begin{remark}
For a bounded test function $|\test|\leq1$, we trivially get $|\is(\test) -
\tm(\test)| \leq 2$;  hence the bounds on bias and MSE provided in Theorem
\ref{non-asymptotictheorem}  are useful only when they are smaller than
2 and 4, respectively. 
\end{remark}
The upper bounds stated in this result suggest that it is good practice
to keep $\rho/N$ small in order to obtain good importance sampling 
approximations. This heuristic dominates the developments in
the remainder of the paper, and in particular our wish to study the
behaviour of $\rho$ in various limits.
 The result trivially
  extends to provide bounds on the mean square error for
  functions bounded by any other known bound different from 1. For
  practical purposes the Theorem is directly applicable to instances where
  importance sampling is used to estimate probabilities, such as in
  rare event simulation. However, its primary role is in providing a
  bound on the particle approximation error, which is naturally
  defined over bounded functions, as is common with weak convergence
  results. It is also important to realise that such a result will not
  hold without more assumptions on the weights for unbounded test
  functions; for example when $g$ has third moment but not fourth
  under $\pi$, then $\tm(g^2)<\infty$, $\prm(g^2)<\infty$ but the
  importance sampling estimator of $\tm(g^2)$ has infinite
  variance. We return to extensions of the Theorem for unbounded test
  functions in section \ref{thm:nonasmom} below.

\subsection{Connections, Interpretations and Extensions}

Theorem \ref{non-asymptotictheorem} clearly demonstrates the role
of $\rho$, the second moment of the target density with respect to the
proposal, in determining the number of samples required to effectively
approximate expectations. 
Here we link $\rho$ to other quantities used in analysis and monitoring of importance sampling algorithms, and we discuss some limitations
of thinking entirely in terms of $\rho.$

\subsubsection{Asymptotic Consistency}

It is interesting to contrast Theorem \ref{non-asymptotictheorem} to a well-known elementary
asymptotic result. First, note that 
\[
\is(\test) - \tm(\test) = {N^{-1} \sum_{n=1}^N {\rn(\un) \over
      \propo(\rn)} \bigr[\test(\un) - \tm(\test)\bigl]  \over N^{-1} \sum_{n=1}^N {\rn(\un) \over
      \propo(\rn)}}\,.
\]
Therefore, under the condition $\propo\bigl(\rn^2\bigr)<\infty$, and provided
additionally that
$\propo\bigl(\rn^2 \test^2\bigr)<\infty$, an application of the Slutsky
lemmas gives that 
\begin{equation}\label{non-asymptoticchopin}
\sqrt{N}\bigl( \is(\test) - \tm(\test) \bigl) \implies \G
\Bigg(0,{\propo\bigl(\rn^2 \testc^2\bigr) \over \propo(\rn)^2}  \Bigg)\,,
\quad \textrm{where}\,\, \testc := \test - \tm(\test)\,.
      \end{equation}
For bounded $|\test|\le 1$, the only condition needed for appealing to
the asymptotic result is $\propo\bigl(\rn^2\bigr) < \infty$. Then \eqref{non-asymptoticchopin} gives that,
for  large $N$ and since $|\testc| \le 2$,
\begin{equation*}
\E \left[\left(\is(\test) - \tm(\test)\right)^2\right]  \lessapprox \frac4N\rho\,,
\end{equation*}
which is in precise agreement with Theorem \ref{non-asymptotictheorem}.

\subsubsection{Effective Sample Size}
\label{ssec:ess}

Many practitioners define the  \emph{effective sample size} by the formula 
\[
\ESS := \left( \sum_{n=1}^N (\wn)^2\right)^{-1} = {\left(\sum_{n=1}^N
    \rn(\un)\right)^2 \over \sum_{n=1}^N \rn(\un)^2} = N {\mc(g)^2
  \over \mc\bigl(g^2\bigr)} \,,
\] where $\mc$ is  the empirical Monte Carlo random measure 
\begin{equation*}
\mc:=\frac{1}{N}\sum_{n=1}^N \delta_{\un}, \quad \un\sim \propo.
\end{equation*}
By the Cauchy-Schwarz inequality it follows that $\ESS \le N$.
Furthermore, since the weights lie in $[0,1]$, we have
$$\sum_{n=1}^N (\wn)^2 \le \sum_{n=1}^N \wn=1$$
so that $\ESS \ge 1.$ These upper and lower bounds may be
attained as follows.
If all the weights are equal, and hence take value $N^{-1}$,
then $\ESS=N$, the optimal situation. On the other hand if exactly $k$
weights take the same value, with the remainder then zero,
$\ESS=k$; in particular the lower bound of $1$ is attained if precisely
one weight takes the value $1$ and all others are zero. 

For large enough $N$, and provided $\propo\bigl(\rn^2\bigr)<\infty$, 
the strong law of large numbers gives
\[
{\ESS} \approx N /\rho\,.
\]
Recalling that $\rho \ge 1$ we see that $\rho^{-1}$ quantifies the
proportion of particles that effectively characterize the sample size,
in the large particle size asymptotic. Furthermore, by Theorem \ref{non-asymptotictheorem}, we
have  that, for large $N$,
\[ \sup_{|\test|\le 1}\E \left[\left(\is(\test) - \tm(\test)\right)^2\right]
\lessapprox  {4 \over \ESS}\,.
\]
This provides a further justification for the use of $\ESS$ as an 
effective sample size, in the large $N$ asymptotic regime.

\subsubsection{Probability Metrics} 
\label{ssec:met}

Intuition tells us that importance sampling will perform
well when the distance between proposal $\prm$ and target $\tm$ 
is not too large.  Furthermore we have shown the role of $\rho$ 
in measuring the rate of convergence of importance sampling. 
It is hence of interest to explicitly link $\rho$ to distance metrics 
between $\prm$ and $\tm.$ In fact we consider asymmetric divergences
as distance measures; these are not strictly metrics, but certainly
represent useful distance measures in many contexts in probability. 
First consider the $\chi^2$ divergence, which satisfies
\begin{equation}
\label{eq:chi}
D_{\chi^2}(\tm\|\propo) := \pi \left( \left[{g \over \pi(g)}-1
  \right]^2 \right) = \rho -1\,.
\end{equation}
%Thus, Theorem \ref{non-asymptotictheorem} suggests that the number of
%particles required for accurate importance sampling scales
%linearly with the $\chi^2$ divergence between proposal
%and target.  
The Kullback-Leibler divergence is given by 
\[ \dkl(\tm\|\propo) := \pi \left( {g \over \pi(g)} \log  {g \over
      \pi(g)} \right)\,,
\]
and may be shown to satisfy 
\begin{equation}
\label{eq:kl}
\rho \geq e^{\dkl(\tm\|\propo) }\,. 
\end{equation}
Thus Theorem \ref{non-asymptotictheorem} suggests that the number of
particles required for accurate importance sampling scales
exponentially with the Kullback-Leibler
divergence between proposal and target and linearly
with the $\chi^2$ divergence.

\subsubsection{Beyond Bounded Test Functions}

In contrast to Theorem \ref{non-asymptotictheorem}, the asymptotic result \eqref{non-asymptoticchopin},
establishes the convergence rate $N^{-1/2}$ (asymptotically) under the
weaker moment assumption on the test function $\propo\bigl(\rn^2
\phi^2 \bigr)<\infty$.    It is thus of interest to derive non-asymptotic 
bounds on the MSE and bias for
much larger classes of test functions. This can be achieved
at the expense of more assumptions on the importance weights. 
The next theorem addresses the issue
of relaxing the class of test functions, whilst still deriving
non-asymptotic bounds. By including the result we also highlight
the fact that, whilst $\rho$ plays an important role in quantifying 
the difficulty of importance sampling, other quantities may be
relevant in the analysis of 
importance sampling for unbounded test functions.  Nonetheless, the sufficiency and necessity of scaling the number of samples with 
$\rho$ is understood in certain settings, as will be discussed in
 the bibliography at the end of this
section.

To simplify the statement we first introduce the following notation.
We write $m_t[h]$ for the $t$-th {\em central moment} with respect to $\propo$ of a function $h:\X\to\R$. That is,
\[m_t[h]:=\propo(|h(u)-\propo(h)|^t ).\]
We also define, as above, $\testc:=\test-\tm(\test).$

\begin{theorem}\label{thm:nonasmom}
Suppose that $\test$ and $\rn$ are such that $C_{\rm MSE}$ defined below is finite:
\begin{align*}
C_{\rm MSE}:=\frac3{\const^2}m_2[\test
\rn]+\frac{3}{\const^4}&\propo(| \test \rn|^{2d})^\frac1d{C}_{2e}^\frac1e m_{2e}[\rn]^\frac1e \\
&+\frac{3}{\const^{2(1+\frac1p)}}\propo(|\test|^{2p})^\frac1p{C}_{2q(1+\frac1p)}^\frac1qm_{2q(1+\frac1p)}[\rn]^\frac1q.
\end{align*}
Then the bias and MSE of importance sampling when applied to
approximate $\mu(\phi)$ may be characterized as follows:
\begin{align*}
\Bigl|\E\bigl[\is(\test)-\tm(\test)\bigr]\Bigr|\leq \frac{1}{N}
\Biggl(\frac{2}{\Zpi^2}m_2[\rn]^\frac12m_2[\testc \rn]^\frac12+2C_{\rm MSE}^\frac12\frac{\propo(\rn^2)^\frac12}{\const}\Biggr)
\end{align*}
and
\begin{align*}
\E\Bigl[\bigr(\is(\test)-\tm(\test)\bigr)^2\Bigr]\leq \frac{1}{N}C_{\rm MSE}.
\end{align*}
The constants ${C}_t>0, t \ge 2$, satisfy $C_t^\frac1t\leq t-1$
and the two pairs of parameters $d,e$, and $p,q$ are conjugate 
pairs of indices satisfying $d,e,p,q \in (1,\infty)$ and $d^{-1}+e^{-1}=1$,
$p^{-1}+q^{-1}=1.$
\end{theorem}

\begin{remark}
In Bayesian inverse problems  $\propo(\rn)<\infty$ often implies that
$\propo\bigl(\rn^s\bigr)<\infty$ for any positive $s$; we will
demonstrate this in a particular case in section \ref{sec:BIP}.  
In such a case, Theorem \ref{thm:nonasmom} combined with H\"older's 
inequality shows that importance sampling converges at rate
$N^{-1}$ for any test function $\test$ satisfying $\propo\bigl(|\test|^{2+\epsilon}\bigr)<\infty$ for some $\epsilon>0.$  
Note, however, that the constant in the ${\cal O}(N^{-1})$ error
bound is not readily interpretable simply in terms of $\rho;$ in particular
the expression necessarily involves moments of $g$ with exponent greater
than two.
\end{remark}

\subsection{Behaviour of the Second Moment $\rho$}

Having demonstrated the importance of $\rho$, the second moment of the
target-proposal density, we now show how it behaves in high dimensional
problems and in problems where there are measure concentration
phenomena due to a small parameter in the likelihood. These two limits
will be of importance to us in subsequent sections of the paper, where
the small parameter measure concentration effect will arise due to high quality
data.

\subsubsection{High State Space Dimension and Absolute Continuity}
\label{ssec:con}

The preceding three subsections have demonstrated how, when
the target is absolutely continuous with respect to the
proposal, importance sampling converges as the square root of
$\rho/N.$ It is thus natural to ask if, and how, this desirable
convergence breaks down for sequences of target and proposal measures which
become increasingly close to singular. 
To this end, suppose that the underlying space  is the
Cartesian product $\R^d$ equipped with the corresponding product
$\sigma$-algebra, the proposal is a product 
measure and the un-normalized weight function also has a product form,
as follows: 
\[
\propo_d(d u) = \prod_{i=1}^d \propo_1(d \ui),\quad \tm_d(du)=
\prod_{i=1}^d \tm_1(d \ui), \quad g_d(u) = 
\exp\left \{  - \sum_{i=1}^d h (\ui) \right\},
\]
for  probability measures $\prm_1,\tm_1$ on $\R$
and $h:\R\to\R^+$ (and we assume it is not
constant to remove the trivial case $\tm_1=\propo_1$).
We index the  proposal, target, density and $\rho$ with respect to $d$
since interest here lies in the limiting behaviour  as $d$
increases. In the setting of \eqref{abscont} we now have
\begin{equation}\label{product}
\tm_d(du)\propto g_d(u)\prm_d(du).
\end{equation}

By construction $g_d$ has all polynomial moments under $\prm_d$ and importance sampling for each $d$ has  the
good properties  developed in the previous sections. 
It is also fairly straightforward to see that $\tm_\infty$ and
$\propo_\infty$ are mutually singular when $h$ is not constant: one 
way to see this is to note that
$$\frac{1}{d} \sum_{i=1}^d u(i)$$
has a different almost sure limit under $\tm_\infty$ and
$\propo_\infty$. Two measures cannot
be absolutely continuous unless they share the same almost sure
properties. Therefore $\tm_{\infty}$ is not absolutely continuous with
respect to $\propo_{\infty}$ and  importance
sampling is undefined in the limit $d=\infty$. As a consequence
we should expect to see a degradation 
in its performance for large state space dimension $d$.

To illustrate this degradation note that 
under the product structure \eqref{product}, we 
have $\rho_d = (\rho_1)^d$.  Furthermore $\rho_1>1$ 
(since $h$ is not constant). Thus $\rho_d$ grows exponentially
with the state space dimension suggesting, when combined with 
Theorem \ref{non-asymptotictheorem}, 
that exponentially many particles are required, 
with respect to dimension, to make importance sampling accurate.

\begin{comment}
A useful perspective on the preceding, which links to our \marginpar{Remove everything until: It is important...}
discussion of the small noise limit in the next subsection, is
as follows.
By the central limit theorem we have that, for large $d$, 
\begin{equation}
\label{eq:gd}
g_d(u) \approx c' \, \exp ( -\sqrt{d} c z \bigr) , \quad z \sim \n(0,1),
\end{equation}
where  $c,c'>0$ are constants with respect to $z$; in addition $c$
is independent of dimension $d$, whilst $c'$ may depend on $d$. 
From this it follows that (noting that any constant scaling, such
as $c'$, disappears from the definition of $\rho_d$)
\begin{equation}\label{rhod}
\rho_d =  \frac{\propo_d(g_d^2)}{\propo_d(g_d)^2}
\approx \frac{\E\exp ( -2\sqrt{d} c z )}{\bigl(\E\exp ( -\sqrt{d} c z)\bigr)^2}\,, %= \frac{\propo(\rn^2)}{{\propo(\rn)}^2} \, 
\end{equation}
where, here, $\E$ denotes expectation with respect to $z \sim \n(0,1).$
Using the fact that $\E e^{-az}=e^{a^2/2}$ we see that $\rho_d
\approx e^{c^2 d}$. 
\end{comment}

It is important to realise that it is not the product structure {\em per se} 
that leads to the collapse, rather the lack of absolute continuity in the 
limit of infinite state space dimension. Thinking about the role of high dimensions
in this way is very instructive in our understanding of high dimensional
problems, but is very much related to the setting in which  
all the coordinates of the problem play a similar role. This does not
happen in many application areas. Often there is a diminishing response
of the likelihood to perturbations in growing coordinate index. When this is the case,
increasing the state space dimension has only a mild effect in the cost of the problem,
and it is possible to have well-behaved infinite dimensional limits; we will see this
perspective in subsections \ref{ssec:gen2}, \ref{ssec:ess2} and 
\ref{ssec:con2} for inverse problems, and 
subsections \ref{ssec:gen3}, \ref{ssec:ess3} and \ref{ssec:con3} 
for filtering.

% However, in Bayesian inverse problems it is common that 
%  In general, it is  
% {\em not}  expected that the dimension of $\X$ directly relates to the
% value of $\rho.$ A paradigmatic example is when the $d$-th marginals
% of measures $\vec{\prm}$ and $\vec{\tm}$ in $\R^d$ become closer as we
% consider a sequence of problems with increasing $d$. This situation
% arises naturally in Bayesian inverse problems and prompts the need to
% define a suitable notion of dimension for these problems. Two notions
% of \emph{effective} dimension and their relationship to absolute
% continuity and the value of $\rho$ are discussed in Section \ref{BIP}.
% We shall see that, even when formulated in infinite dimensional
% Hilbert spaces, Bayesian inverse problems often have finite
% {effective} dimension. 

\subsubsection{Singular Limits}
\label{ssec:sin}

In the previous subsection we saw an example where for high dimensional state spaces the target and
proposal became increasingly close to being mutually singular, resulting
in $\rho$ which grows exponentially with the state space dimension. 
In this subsection we observe that mutual singularity
can also occur because of small parameters in the unnormalized
density $g$ appearing in \eqref{abscont}, even in problems of fixed
dimension; this will lead to $\rho$ which grows algebraically
with respect to the small parameter. 
To understand this situation let $\X=\R$ 
and consider \eqref{abscont} in the setting where
$$g_\epsilon(u)=\exp\bigl(-\epsilon^{-1}h(u)\bigr)$$ 
where $h:\R \to \R^+$.  We will write $g_\epsilon$ and 
$\rho_\epsilon$ to highlight the dependence of these quantities on
$\epsilon$.
Furthermore assume, for simplicity, that $h$ is twice differentiable and 
has a unique minimum at $u^\star,$ and that $h''(u^\star)>0$. Assume,
in addition, that 
$\prm$ has a Lebesgue density with bounded first derivative. Then the
Laplace method shows that
$$\E \exp\bigl(-2\epsilon^{-1}h(u)\bigr) \approx \exp\bigl(-2\epsilon^{-1}h(u^\star)\bigr)\sqrt\frac{2\pi\epsilon}{2h''(u^\star)}$$
and that
$$\E \exp\bigl(-\epsilon^{-1}h(u)\bigr) \approx \exp\bigl(-\epsilon^{-1}h(u^\star)\bigr)\sqrt\frac{2\pi\epsilon}{h''(u^\star)}.$$
It follows that
$$\rho_\epsilon \approx \sqrt \frac{h''(u^\star)}{4\pi \epsilon}.$$
Thus Theorem \ref{non-asymptotictheorem} indicates that
the number of particles required for importance sampling to be
accurate should grow at least as fast as $\epsilon^{-\frac12}.$

\subsection{Discussion and Connection to Literature}
\label{ssec:lit}

\subsubsection{Metrics Between Random Probability Measures}
In subsection \ref{ssec:gen} we introduced the importance sampling
approximation of a target $\tm$ using a proposal $\prm$, both related
 by \eqref{abscont}. The resulting
particle approximation measure $\is$ is random because it is based on samples 
from $\prm.$ Hence $\is(\phi)$ is a 
\emph{random} estimator of $\tm(\phi).$ 
This estimator is in general biased, and therefore a reasonable 
metric for its quality is the MSE
\[
\E\Bigr[\bigl(\is(\phi) - \tm(\phi)\bigr)^2\Bigr]\,,
\]
where the expectation is with respect to the randomness in the 
measure $\is.$ We bound the MSE over the class of bounded
test functions in Theorem \ref{non-asymptotictheorem}.
In fact we may view this theorem as giving a bound on a distance
between the measure $\tm$ and its approximation $\is.$ 
To this end let $\nu$ and $\mu$ denote mappings from an underlying 
probability space (which for us will be that associated with $\prm$) 
into the space of probability measures on $(\X, \F)$; in the
following, expectation $\E$ is with respect to this underlying probability space. 
In \cite{RH13} a distance $d(\cdot,\cdot)$ beween such random measures is defined by
\begin{align}\label{vhmetric}
d(\nu,\tm)^2=\sup_{|\test|\leq 1}\E \Bigl[ \bigl(\nu(\test)-\tm(\test)\bigr)^2\Bigr].
\end{align} 
The paper \cite{RH13} used this distance to study the convergence of particle
filters. 
Note that if the measures are not random the distance reduces to
total variation. Using this distance, together with the
discussion in subsection \ref{ssec:met} linking $\rho$ to the
$\chi^2$ divergence, 
we see that Theorem \ref{non-asymptotictheorem} states
that
$$d(\is,\tm)^2 \le \frac{4}{N}\bigl(1+D_{\chi^2}(\tm\|\propo)\bigr).$$
In subsection \ref{ssec:met} we also link $\rho$ to the Kullback-Leibler divergence; the bound \eqref{eq:kl} can be 
found  in  Theorem 4.19 of \cite{GaborConc}.
As was already noted, this suggests the need to increase
the number of particles linearly with $D_{\chi^2}(\tm\|\propo)$ or exponentially with $\dkl(\tm\|\propo).$ 
\begin{comment}
Although Theorem \ref{non-asymptotictheorem} constitutes an insightful synthesis of many ideas, objects and results on importance sampling, the analysis is restricted to bounded test functions. We dispense with this restriction in Theorem \ref{thm:nonasmom} at the expense of losing some of the clarity in the aforementioned result.
\end{comment}

\subsubsection{Complementary Analyses of Importance Sampling Error}
Provided that $\log\bigl(\frac{g(u)}{\propo(g)}\bigr), \, u\sim \tm,$  is concentrated around its expected value, as often happens in large dimensional and singular limits, it has recently been shown \cite{CP15} that using a sample size of approximately $\exp\bigl( \dkl(\tm\|\propo)\bigr)$ is both necessary and sufficient in order to control the $L^1$ error $\mathbb{E}|\is(\phi) - \tm(\phi)|$ of the importance sampling estimator $\is(\test)$. Theorem \ref{non-asymptotictheorem} is similar to
\cite[Theorem 7.4.3]{del2004feynman}. However the later result uses
a metric defined over subclasses of bounded functions. The
resulting constants in their bounds rely on covering numbers, 
which are often  intractable. In contrast, the constant $\rho$ in 
Theorem \ref{non-asymptotictheorem} is more amenable to analysis and
has several meaningful interpretations as we highlight in this paper.
%that will be explored in the
%remainder of the paper, including the one resulting in the preceding
%display. 
The central limit result in equation \eqref{non-asymptoticchopin} 
shows that for large $N$ the upper bound in Theorem \ref{non-asymptotictheorem} is sharp. Equation \eqref{non-asymptoticchopin}
can be seen as a trivial application of deeper  central limit theorems
for particle filters, see \cite{chopin2004central}. 

This discussion serves to illustrate the fact that 
a universal analysis of importance
sampling in terms of $\rho$ alone is not possible.
Indeed Theorem \ref{thm:nonasmom} shows that the
expression for the error constant in useful error bounds
may be quite complex when considering test functions
which are not bounded.
The constants ${C}_t>0$, $t\geq 2$ in Theorem \ref{thm:nonasmom}
are determined by the Marcinkiewicz-Zygmund  inequality \cite{RL01}. 
The proof follows the approach of \cite{DL09} for evaluating moments of ratios. 
Despite the complicated dependence of error constants on the problem
at hand, there is further evidence for the centrality of the second
moment $\rho$ in the paper \cite{dan:new}. There it is shown 
(see Remark 4) that, when $\rho$ 
is finite, a {\em necessary} condition for accuracy within the class
of functions with bounded second moment under the proposal, is that
the sample size $N$ is of the order of the $\chi^2$ divergence, and hence of the order of
$\rho.$

Further importance sampling results have been proved within the study of convergence properties of various versions of the particle filter as a numerical method for the approximation of the true filtering/smoothing distribution.  These results are often formulated  in finite dimensional state spaces, under bounded likelihood assumptions and for bounded test functions, see \cite{CDL98}, \cite{DMM00}, \cite{CD02}, \cite{MCD13},  \cite{ACMR14}. Generalizations for continuous time filtering can be found in \cite{BC09} and \cite{H13}.

\subsubsection{Effective Sample Size, and the Case of Infinite Second Moment}
The effective sample size $\ESS$, introduced in subsection \ref{ssec:ess},
is a standard statistic used to assess and monitor particle
approximation errors in importance sampling  \cite{kong1992note}, \cite{kong1994sequential}. The effective sample size $\ESS$ does not depend on any specific test function, but is rather a particular function of the normalized weights which quantifies their variability. So does $\rho,$ and as we show
in subsection \ref{ssec:ess} there is an asymptotic connection between both. 
Our discussion of $\ESS$ relies on the condition 
$\propo\bigl(\rn^2\bigr) < \infty$. Intuitively,  
the particle approximation will be rather poor when
this condition is not met. Extreme value theory provides
some clues about the asymptotic particle approximation error.  
First it may be shown that, 
regardless of whether $\propo\bigl(\rn^2\bigr)$ is finite or not,
but simply on the basis that $\propo(\rn) < \infty$, the largest
normalised weight, $\wN$, will
converge to 0 as $N \to \infty$; see for example section $3$ of
\cite{downey} for a review of related results. On the other hand,
\cite{mcleish} shows that, for large $N,$
\[
\E\left[ {N \over \ESS} \right] \approx  \int_0^N \gamma S(\gamma)
d \gamma,
\]
where $S(\gamma)$ is the survival function of the distribution of the
un-normalized weights, $\gamma := \rn(u)$ for $u \sim \propo$. For
instance, if the weights have density proportional to 
$\gamma^{-a-1}$, for $1 < a <2$, then $\propo\bigl(\rn^2\bigr)=\infty$ and, for large enough $N$ and constant $C$, 
\[
\E\left[ {N \over \ESS} \right] \approx  C\, N^{-a+2}\,.
\]
Thus, in contrast to the situation where $\propo\bigl(\rn^2\bigr) < \infty$, in this setting the effective sample size does not 
grow linearly with $N$.

\subsubsection{Large State Dimension, and Singular Limits}

%In subsections \ref{ssec:con} and \ref{ssec:sin} we studied how 
%limits in which the target and proposal become closer and closer
%to being mutually singular (breakdown of absolute continuity)
%lead to problems for importance sampling. 
In subsection \ref{ssec:con}
we studied high dimensional problems with a 
product structure that  enables analytical 
calculations. The use of such product structure was pioneered for
MCMC methods in \cite{GGR96}.
It has then been recently employed in the analysis 
of importance sampling in high nominal dimensions,
starting with the seminal paper \cite{BBL08}, and leading on to others 
such as \cite{beskos2014stability}, \cite{beskos2014stable},
 \cite{BLB08}, \cite{snyder2008obstacles}, \cite{snyder2011alone}, \cite{slivinskipractical}, and \cite{snyderberngtsson}.

In \cite[Section 3.2]{BBL08} it is shown that the maximum normalised importance sampling
weight can be approximately written as
\[
\wN \approx {1 \over 1 + \sum_{n>1} \exp\bigl\{ - \sqrt{d} c (z^{(n)}-z^{(1)})
  \bigr\}},
\]
where $\{z^n\}_{n=1}^N$ are samples from $\n(0,1)$ and the $z^{(n)}$ are the ordered statistics. 
In \cite{BLB08} a direct but non-trivial calculation shows that if $N$
does not grow exponentially with $d$, the sum in the
denominator converges to 0 in probability and as a result the maximum
weight to 1. Of course this means that all other weights are
converging to zero, and that the effective sample size is $1$.
It chimes with the heuristic derived in subsection \ref{ssec:con}
where we show that $\rho$ grows exponentially with $d$ and that
choosing $N$ to grow exponentially is thus necessary to keep the
upper bound in Theorem \ref{non-asymptotictheorem} small. 
The phenomenon is an instance of what is sometimes termed 
{\em collapse of importance sampling} in high dimensions. This type of
behaviour can be obtained for other classes of targets
and proposals; see \cite{BBL08}, \cite{snyder2008obstacles}. 
Attempts to alleviate this behaviour include the use of tempering \cite{beskos2014stability} or combining importance sampling with Kalman-based algorithms \cite{FK13}. However, the range of applicability of these ideas is still to be studied.
In subsection \ref{ssec:sin} we use the Laplace method. 
This is a classical methodology for approximating integrals and can be found in many text books; 
see for instance \cite{bender1999advanced}.

\section{Importance Sampling and Inverse Problems}
\label{sec:BIP}

The previous section showed that the distance 
between the proposal and the target is key in understanding the 
computational cost of importance sampling and the central role
played by $\rho$.
In this section we study the computational cost
of importance sampling applied in the context of
Bayesian inverse problems, where the target will be the posterior and the proposal the prior. 
To make the analysis tractable we consider linear Gaussian inverse problems, but our ideas extend beyond this setting.
Subsection \ref{ssec:gen2} describes the setting and necessary background on inverse problems. Then subsection \ref{ssec:ess2} introduces various notions of ``intrinsic dimension'' for linear Gaussian inverse problems;
a key point to appreciate in the sequel is that this dimension
can be finite even when the inverse problem is posed in an infinite
dimensional Hilbert space. The analysis of importance sampling starts in subsection \ref{ssec:con2}. The main result is Theorem \ref{maintheorem}, that shows the equivalence between (i) finite {\em intrinsic} dimension, (ii) absolute continuity of the posterior (target) with respect to the prior (proposal), and (iii) the central quantity $\rho$ being finite. 
The section closes with a thorough study of singular limits in subsection \ref{ssec:sin2} and a literature review in subsection \ref{ssec:lit2}.

\subsection{General Setting}
\label{ssec:gen2}

We study the inverse problem of finding $u$ from $y$ where
\begin{equation}\label{invprob}
y=\fp u+\eta.
\end{equation}
In particular we work in the setting where $u$ is an element of
the (potentially infinite dimensional) 
separable Hilbert space $(\h,\pr{\cdot}{\cdot},\norm{\cdot})$.
Two cases will help guide the reader:

 \begin{example}[Linear Regression Model]\label{ex:lm}
In the context of the linear regression model, $u \in  \R^{d_u}$ is the regression
  parameter vector,  $ y \in \R^{d_y}$ is a vector of training outputs
  and $K \in \R^{d_y \times d_u}$ is the  so-called design matrix whose
  column space is used to construct a linear predictor for  the
  scalar output. In this setting, $d_u, d_y < \infty$, although in modern
  applications both might be
very large, and the case $d_u \gg d_y$ 
is the so-called ``large $p$ (here $d_u$) small $N$ (here $d_y$)''
  problem.  %The Bayesian approach we undertake below becomes the Gaussian conjugate Bayesian analysis of linear regression models, a la \cite{LS72}. 
 \end{example}

 \begin{example}[Deconvolution Problem]\label{ex:lm1}
In the context of signal deconvolution, $u\in L^2(0,1)$ is a square
integrable unknown signal on the unit interval, $K:L^2(0,1)\to
L^2(0,1)$ is a convolution operator $Ku(x)=(\phi\star
u)(x)=\int_0^1\phi(x-z)u(z)dz$, and $y=Ku+\eta$ is the noisy
observation of the convoluted signal where $\eta$ is observational
noise. The convolution kernel $\phi$ might be, for example, a Gaussian
kernel $\phi(x)=e^{-\delta x^2}$.
Note also that discretization of the deconvolution problem will
lead to a family of instances of the preceding linear regression
model,
parametrized by the dimension of the discretization space. 
 \end{example}

 The infinite dimensional setting does require 
some technical background, and this is outlined in the Supplementary Material. Nevertheless, the reader versed only in
finite dimensional Gaussian concepts will readily make sense
of the notions of intrinsic dimension described in 
subsection \ref{ssec:ess2} simply by thinking of 
(potentially infinite dimensional)
matrix representations of covariances.

In equation \eqref{invprob} the data $y$ is comprised of the
image of the unknown $u$ under a linear map $K$, with added
observational noise $\eta.$ 
Here $K$ can be formally thought of as being a bounded linear operator 
in $\h$, which is ill-posed in the sense that if we attempt to invert 
the data using the (generalized) inverse of $K$, we get amplification of 
small errors $\eta$ in the observation to large errors in the 
reconstruction of $u$. In such situations, we need to use regularization techniques in order to stably reconstruct the unknown $u$ from the noisy data $y$. %The ill-posedness of $K$ is made precise in
%Assumption \ref{Generalinvprobassumption} below. 

%We work in a scale of spaces with basic Hilbert space $\X$, for example $\X=L^2[0,1]$ and the scale of spaces is the Sobolev scale. We assume that the unknown $u$ and the observational noise $\eta$, live in the spaces $\tX$ and  $\Y$ respectively, which belong in the assumed scale of spaces. The forward operator $K$ is viewed as a bounded operator mapping $\tX$ to $\Y$.
%In order for the inverse problem to be ill-posed, in the sense that the (generalized) inverse of $K$ is unbounded from the observation space we need $\Y\subset\tY$.
%For more details on the spatial setup see Remark \ref{rem:spaces} below.

We assume Gaussian observation noise  $\eta\sim\noise:=\G(0, \, \nc)$ and adopt a Bayesian approach by putting a prior  on the unknown $u \sim\prior=\G(0,\, \prc)$. Here and throughout  $\nc:\h\to\h$ and $\prc:\h\to\h$ are bounded, self-adjoint, positive-definite linear operators. Note that we do not assume that $\nc$ and $\prc$ are trace class, which introduces some technical difficulties since $\eta$ and $u$ do not necessarily live in $\h$. This is discussed in the Supplementary Material. 

%As discussed in \ds{subsection} \ref{ssec:Gaussian},
%if the covariance $\nc$ (respectively $\prc$) is trace class then \marginpar{need to fix the two references to supplementary material}
%$\eta \sim \noise$ (respectively $u \sim \prior$) is almost surely
%in $\h$. On the other hand, as also discussed in 
%\ds{subsection} \ref{ssec:Gaussian}, when \ds{the} covariance $\nc$
%(respectively $\prc$) is not 
%trace-class we have that $\eta \notin \h$ but  $\eta\in \Y$ $\noise$-almost
%surely (respectively $u\notin \h$ but $u \in \X$ $\prior$-almost surely) 
%where $\Y$ (respectively $\X$) strictly contains $\h$; 
%indeed $\h$ is compactly embedded into $\X,\Y.$

\begin{comment}
We assume Gaussian observation noise  $\eta\sim\noise:=\G(0, \, \nc)$ and adopt a Bayesian approach by putting a prior  on the unknown $u \sim\prior=\G(0,\, \prc)$, where  $\nc:\h\to\h$ and $\prc:\h\to\h$ are bounded, self-adjoint, positive-definite linear operators. 
As discussed in subsection \ref{ssec:Gaussian},
if the covariance $\nc$ (respectively $\prc$) is trace class then 
$\eta \sim \noise$ (respectively $u \sim \prior$) is almost surely
in $\h$. On the other hand, as also discussed in 
subsection \ref{ssec:Gaussian}, when covariance $\nc$
(respectively $\prc$) is not 
trace-class we have that $\eta \notin \h$ but  $\eta\in \Y$ $\noise$-almost
surely (respectively $u\notin \h$ but $u \in \X$ $\prior$-almost surely) 
where $\Y$ (respectively $\X$) strictly contains $\h$; 
indeed $\h$ is compactly embedded into $\X,\Y.$
\end{comment}

The Bayesian solution is the posterior distribution $u|y\sim\post$. In the finite dimensional setting the prior $\prior$ and the posterior $\post$ are Gaussian
conjugate and $\post=\G(\posm,\posc)$, with mean and covariance given
by  
\begin{align}\label{eq:posm}
\posm=\prc\fp^\ast (\fp \prc\fp^\ast +\nc)^{-1}y,
\\\label{eq:posc}
\posc=\prc-\prc\fp^\ast (\fp\prc\fp^\ast +\nc)^{-1}\fp\prc.
\end{align}
A simple way to derive the expressions above is by working with precision matrices.  Indeed, using Bayes' rule 
and completion of the square gives
\begin{align}\label{eq:fdposc}
\posc^{-1}&=\prc^{-1}+ \fp^\ast  \nc^{-1} \fp,\\ \label{eq:fdposm}
\posc^{-1}\posm&= \fp^\ast  \nc^{-1} y.\end{align}
An application of Schur complement then yields \eqref{eq:posm} and \eqref{eq:posc}.

\begin{remark}\label{remarkfininf}
Under appropriate conditions ---see the references in the literature review subsection \ref{ssec:lit2} and the Supplementary Material--- formulae \eqref{eq:posm}-\eqref{eq:fdposm} can be established in the infinite dimensional setting. From now on and whenever necessary we \emph{assume} that these expressions are available in the general Hilbert space setting that we work in. In particular Proposition \ref{prop:new} makes use of the formula \eqref{eq:fdposc} for the posterior precision.
\end{remark}

Under the prior and noise models 
we may write $u=\Sigma^{\frac12}u_0$ 
and $\eta=\Gamma^{\frac12}\eta_0$ where $u_0$ and $\eta_0$ are
independent centred Gaussians with identity covariance operators
(white noises).
Thus, we can write \eqref{invprob}, for $y_0=\Gamma^{-\frac12}y$, as
\begin{equation}
\label{eq:sip}
y_0=Su_0+\eta_0, \quad S=\nc^{-\frac12}\fp\prc^\frac12.
\end{equation}
Therefore all results may be derived for this inverse problem, and translated
back to the original setting. This intuition demonstrates the centrality of the operator $S$ linking $\fp, \prc$ and $\nc.$ The following assumption will be in place in the remainder of the paper. 
\begin{comment}
We tacitly assume that $\fp$ can be extended to act on elements in $\X$ 
and that the sum of $\fp u$ and $\eta$ makes sense in $\Y$.
 In the setting outlined above we assume that the prior acts as a regularization for the inversion of the data $y$. %now make precise the assumption that  $\fp$ is ill-posed. We assume that  $K$ viewed as a linear operator  from the support of the prior to the support of the noise, is unboundedly invertible. 
This is encoded in  the following assumption 
on the relationship between the operators $\fp, \prc$ and $\nc$. 
\end{comment}

\begin{assumption}\label{Generalinvprobassumption}
Define $\auxm=\nc^{-\frac12}\fp\prc^\frac12$, $A=\auxm^\ast \auxm$
and assume that $A,$ viewed as a linear operator in $\h$, is bounded. Furthermore, assume that the spectrum of $A$ consists of a countable number of eigenvalues, sorted without loss of generality in a non-increasing way: \[\lambda_1\geq \lambda_2\geq\dots\geq\lambda_j\geq\dots\geq0.\]
\end{assumption}

In section \ref{ssec:lit2} we give further intuition on the centrality of the operator $S$ and hence $A$, and discuss the role of the  assumption in the context
of inverse problems. 

\begin{comment}
In subsection \ref{ssec:lit2} we give an intuitive explanation for the centrality
of $A$ and $S$, and discuss the role of the  assumption in the context
of inverse problems. 
\end{comment}

\subsection{Intrinsic Dimension}
\label{ssec:ess2}
 Section \ref{sec:IS}
demonstrates the importance of the distance between the target (here the posterior)
and the proposal (here the prior) in the performance of importance sampling. In the Gaussian setting considered
in this section any such distance is characterized in terms of means and covariances. We now show that the ``size'' of the 
operator $A$ defined in Assumption \ref{Generalinvprobassumption} can be used to quantify the distance between the prior and the posterior covariances, $\Sigma$ and $C$. In subsections \ref{ssec:con2} and \ref{ssec:sin2} we will see that, although $A$ does not contain explicit information on the prior and posterior means, its size largely determines the computational 
cost of importance sampling.

\begin{comment}
The operator $A$ defined in Assumption \ref{Generalinvprobassumption}
plays an important role in what follows because it measures the
size of the difference between the prior and posterior covariances
$\Sigma$ and $C$. The developments in section \ref{sec:IS}
indicate that a key measure determining the computational cost
of importance sampling is the distance between the target (here the posterior)
and the proposal (here the prior). In the Gaussian setting considered
in this section the differences between posterior and prior
covariances will contribute to this distance and we now develop
this idea. Note, however, that we say nothing here about the differences
between prior and posterior means.
\end{comment}

\begin{comment}
We illustrate the ideas in finite state/data space dimensions in the
first instance.   For extensions to 
Hilbert spaces, see the Supplementary Material.
\end{comment}

\begin{prop}
\label{prop:new}
Under the general setting of subsection \ref{ssec:gen2} the following identities hold
$$\tr\bigl((C^{-1}-\Sigma^{-1})\Sigma\bigr)=\tr(A),
\quad \tr\bigl((\Sigma-C)\Sigma^{-1}\bigr)=\tr\bigl((I+A)^{-1}A\bigr).$$
\end{prop}

Thus the traces of $A$ and of $(I+A)^{-1}A$ measure
the differences between the posterior and prior
precision and covariance operators, respectively, relative
to their prior values. For this reason they provide useful measures of
the computational cost of importance sampling, motivating 
the following definitions: 
\begin{equation}\label{eq:efd}
\tau := \tr(A), \quad\quad \effd:=\tr\bigl(( I+A)^{-1}A\bigr)\,.
\end{equation} 
Note that the trace calculates the sum of the eigenvalues and is
well-defined, although may be infinite, in the Hilbert space setting.
We refer to $\effd$ as effective dimension; both $\tau$ and $\effd$
are measures of the intrinsic dimension of the inverse problem at hand.
The next result shows that the intrinsic dimension $\effd$ has the appealing property of being bounded above by the nominal dimension.

\begin{prop}
\label{prop:efd}  
Let $\auxm$ and $A$  be defined as in 
Assumption \ref{Generalinvprobassumption}, and consider the finite dimensional
setting with the notation introduced in Example \ref{ex:lm}. 
\begin{enumerate}
\item The matrices $\Gamma^{1/2} \auxm ( I + A)^{-1} \auxm^*
  \Gamma^{-1/2} \in \R^{d_y\times d_y}$,  $\auxm(  I +A)^{-1}\auxm^\ast\in\R^{d_y\times d_y}$ and $( I+A)^{-1}A\in
    \R^{d_u \times d_u}$ have the same non-zero eigenvalues and hence
the same trace. 
\item  If $\lambda_i >0 $ is a non-zero eigenvalue
of $A$ then these three matrices have corresponding eigenvalue 
$\lambda_i(1+\lambda_i)^{-1}<1$, and 
\[
\effd = \sum_{i} {\lambda_i \over 1+ \lambda_i}  \leq  \nd=\min\{d_u,d_y\}\,.
\]
\end{enumerate}
\end{prop}

Here, recall,  $\nd=\min\{d_u,d_y\}$ is 
the nominal dimension of the problem.
Part 2 of the preceding result demonstrates the connection between
$\effd$ and the physical dimensions of the unknown and observation spaces, 
whilst part 1 demonstrates the equivalence between the traces of
a variety of operators, all of which are used in the literature;
this is discussed in greater detail in
 subsection \ref{ssec:lit2}.
In the Hilbert space setting, recall, the intrinsic dimensions $\effd$ and $\tau$ can be infinite. 
It is important to note, however,
that this cannot happen if the rank of $K$ is finite. That is, the intrinsic dimension $\effd$ (and, as we now show, also $\tau$) is finite whenever the unknown $u$ or the data $y$ live in a finite dimensional subspace of $\h.$ 
The following result relates $\effd$ and $\tau.$ It shows in particular that in the infinite dimensional setting they are finite, or otherwise, at the same time.

\begin{lemma}\label{propequivefd}
Under the general setting of subsection \ref{ssec:gen2}, the operator  $A$ is trace class if and only if 
$(I + A)^{-1} A$ is trace class. Moreover,
the following inequalities hold
\[
{1 \over \| I + A\| } \tr(A) \leq \tr \left(( I + 
A )^{-1} A \right)  \leq 
\tr(A).
\]
As a consequence
\begin{equation}\label{boundtaueffdtau}
{1 \over \| I + A\| } \tau \le \effd\le \tau.
\end{equation}
\end{lemma}

We are now ready to study the performance of importance sampling with posterior 
as target and  prior as proposal. In subsection \ref{ssec:con2} we
identify conditions under which we can guarantee that  $\rho$ 
in Theorem \ref{non-asymptotictheorem} is finite and absolute continuity holds. 
In subsection \ref{ssec:sin2} we then study the growth of 
$\rho$ as mutual singularity is approached in different regimes.
The intrinsic dimensions $\tau$ and $\effd$ will be woven into these
developments.

\subsection{Absolute Continuity}
\label{ssec:con2}

In the finite dimensional setting the Gaussian proposal and target 
distributions have densities with respect to the Lebesgue measure. They are
hence mutually absolutely continuous and it is hence  straightforward to find 
the Radon-Nikodym derivative of the target with respect to the proposal by 
taking the ratio of the respective Lebesgue densities once the posterior
is identified via Bayes' theorem; this gives:
\begin{equation}\label{eq:fdrn}\frac{d\tm}{d\prm}(u)=\frac{d\post}{d\prior}(u;y)\propto\exp\left(-\frac{1}{2}
    u^\ast \fp^\ast \nc^{-1}Ku+ u^\ast \fp^\ast
    \nc^{-1}y\right)=:\rnip(u;y).\end{equation}

Direct calculation shows that, for $d_u, d_y<\infty$ 
 the ratio $\rho$ defined 
in \eqref{rhodef} is finite, and indeed that
$g$ admits all polynomial moments, all of which are positive. 
In this subsection we  study $\rho$ in the Hilbert space setting. In
general there is no guarantee that the posterior is absolutely continuous 
with respect to the prior; when it is not, $g$, and hence $\rho$, are not
defined. We thus seek conditions under which such absolute continuity
may be established.

To this end, we define the likelihood measure $y|u\sim\like:=\G(\fp u,\nc),$ and the joint distribution of $(u,y)$ under the model $\nu(du,dy):=\like(dy|u)\prior(du)$, recalling that $\prior=\G(0,\Sigma).$ We also define the marginal distribution of the data under the joint distribution, 
$\mar(dy)=\rp_y(dy)$. We have the following result: 
\begin{theorem}\label{maintheorem}
Let Assumption \ref{Generalinvprobassumption} hold
and let $\tm=\post$ and $\propo=\prior.$
The following are equivalent:
\begin{enumerate}[i)]
\item $\effd<\infty;$
\item $\tau<\infty;$
\item $\nc^{-1/2}\fp u\in \h$, $\propo$-almost surely;
\item for $\mar$-almost all $y$, the posterior $\tm$ is well defined as 
a measure in any space of full prior measure  and is absolutely continuous with respect to the prior with
%a measure in $\X$  and is absolutely continuous with respect to the prior with 
\begin{equation}\label{rninfdim}
\frac{d\tm}{d\propo}(u)\propto\exp\left(-\frac{1}{2}\norm{\nc^{-1/2}\fp u}^2+\frac12\pr{\nc^{-1/2}y}{\nc^{-1/2}\fp u}\right)=:\rnip(u;y),
\end{equation} 
where $0<\propo\bigl(\rnip(\cdot;y)\bigr)<\infty.$
\end{enumerate}
\end{theorem}

\begin{remark}\label{rem:main}
Due to the exponential structure of $g$, we have that assertion $(iv)$ of the last theorem is immediately equivalent to $g$ being $\nu$-almost surely positive and finite and for $\mar$-almost all $y$ the second moment of the
target-proposal density is finite:
   $$\rho=\frac{\propo\left(\rnip(\cdot;y)^2\right)}{\propo\bigl(\rnip(\cdot;y)\bigr)^2}<\infty.$$
\end{remark}

Item (iii) is a requirement on the regularity of the forward image of draws from the prior, relative to the regularity of the noise.  This regularity condition heavily constrains the space of possible reconstructions and is thus related to the intrinsic dimension of the inverse problem, as we establish here. For a discussion on the regularity of draws from Gaussian measures in Hilbert spaces see the Supplementary Material.
\begin{comment}
Note that item (iii) can also be interpreted as quantifying the dimension of the problem, since it is a requirement on the regularity of the forward image of the unknown, relative to the noise; such regularity condition typically relates to smoothness of the underlying field, and thus to intrinsic dimension, as we show here.
\end{comment}

We have established something very interesting: there
are meaningful notions of intrinsic dimension for inverse problems formulated 
in infinite state/data state dimensions and, when the intrinsic dimension is
finite, importance sampling may be possible as there is absolute continuity;  moreover, in such a situation $\rho$ is finite.
Thus, under any of the equivalent conditions (i)-(iv), 
Theorem $\ref{non-asymptotictheorem}$ can be used to 
provide bounds on the effective sample size 
$\ESS,$  defined in subsection \ref{ssec:ess}; indeed
the effective sample size is then proportional to $N$.

It is now of interest to understand how $\rho$, and the intrinsic
dimensions $\tau$ and $\effd$, depend on various parameters,
such as small observational noise or the dimension of finite dimensional
approximations of the inverse problem.  Such questions are 
studied in the next subsection.

\begin{comment}
\begin{remark}
\modd{Daniel to add a remark saying that connection of $\tau$ and $\effd$ when the posterior exists and is given by... is a consequence of the equivalence between the prior and the posterior, since both are expectations of $Ku$ under these measures.}
\end{remark}
\end{comment}

\subsection{Large Nominal Dimension and Singular Parameter Limits} 
\label{ssec:sin2}

The parameter $\rho$ is a complicated nonlinear function of the eigenvalues of $A$ and the data $y$. 
However, there are some situations in which we can lower bound  $\rho$ in 
terms of the intrinsic dimensions $\tau$, $\effd$ and the size of the 
eigenvalues of $A$. We present
two classes of examples of this type. The first is 
a simple but insightful example in which the eigenvalues cluster into
a finite dimensional set of large eigenvalues and a set of small
remaining eigenvalues. The second involves asymptotic considerations 
in a simultaneously diagonalizable setting.

\subsubsection{Spectral Jump} \label{spectraljumpexample}
Consider the setting where $u$ and $y$ both live in finite dimensional
spaces of dimensions $d_u$ and $d_y$ respectively.
Suppose that $A$ has eigenvalues $\{\lambda_i\}_{i=1}^{d_u}$ with $\lambda_i=C\gg 1$ for $1\le i \le k$, and $\lambda_i\ll 1$ for $k+1\le i \le d_u$;
indeed we assume that
$$\sum_{i=k+1}^{d_u} \lambda_i \ll 1.$$
Then $\tau(A) \approx Ck$, whilst the effective dimension 
satisfies $\effd\approx k$. Using the
identity
\begin{equation*}
2\dkl(\post\| \prior) = \log\Bigl(\det (I+ A)\Bigr) - \tr\Bigl((I+A)^{-1}A\Bigr) +  \posm^* \prc^{-1} \posm \,.
\end{equation*}
and studying the asymptotics for fixed $m$, with $k$ and $C$ large, we
obtain
$$\dkl(\post||\prior) \approx \frac{\effd}{2} \log(C) \, .$$
Therefore, using \eqref{eq:kl},
$$\rho \gtrapprox C^{\frac{\effd}{2}} \,.$$
This suggests that $\rho$ grows exponentially with the \emph{number} of 
large eigenvalues, whereas it has an algebraic dependence on 
the \emph{size} of the eigenvalues.  Theorem $\ref{non-asymptotictheorem}$ 
then suggests that the number of
particles required for accurate importance sampling will grow
exponentially with the number of large eigenvalues, and
algebraically with the size of the eigenvalues. 
A similar distinction 
may be found by comparing the behaviour of
$\rho$ in large state space dimension in subsection \ref{ssec:con} (exponential)
and with respect to small scaling parameter in 
subsection \ref{ssec:sin} (algebraic).

\subsubsection{Spectral Cascade}\label{sssec:spectralcascade}
We  now introduce a three-parameter family of inverse problems, defined through the eigenvalues of $A.$ These three parameters represent the regularity of the prior and the forward map, the size of the observational noise, and the number of positive eigenvalues of $A,$ which corresponds to  the nominal dimension. We are interested in investigating the performance of importance sampling, as quantified by $\rho$, in different regimes for these parameters. We work in the framework of Assumption \ref{Generalinvprobassumption}, and under the following additional assumption:

\begin{comment}

\begin{assumption}\label{ass1}
Within the framework of Assumption $\ref{Generalinvprobassumption}$ we consider the inverse problem
$$y= \fp u + \eta$$ with $u\sim N(0,\prc)$ and $\eta\sim N(0,\gamma I)$. To simplify the analysis we precondition the problem, as detailed in subsection \ref{ssec:lit2} and consider the problem of equivalent intrinsic dimension $$y_0=w_0+\eta_0$$
with $w_0\sim N(0,SS^*)$ and $\eta_0 \sim N(0,  I).$ We assume that the eigenvalues of $A$, which agree with those of $SS^*,$ are given by $\Big\{\frac{j^{-\beta}} {\gamma}\Big\}_{j=1}^\infty.$ We assume that we have data generated from a fixed underlying truth $u^\dagger\equiv \bigl(u^\dagger(1), u^\dagger(2),\ldots\bigr) \in \h.$ 
We consider truncations of this problem to level $d$, where the unknown is given by $u_d=\bigl(u(1), \ldots, u(d)\bigr)$ and has prior covariance $S_d S_d^*$ with  eigenvalues $\Big\{\frac{j^{-\beta}} {\gamma}\Big\}_{j=1}^d.$ The truncated problem has data generated from the truncated truth, given by  $y(j)= u^\dagger(j) + \eta(j)$, where $\eta(j)\sim N(0,1)$ and $1\le j \le d.$ Each of these inverse problems has associated operator $A=A(\beta,\gamma,d)$, with $d\in \N\cup \infty.$
\end{assumption}
\end{comment}
\begin{assumption}\label{ass1}
Within the framework of Assumption \ref{Generalinvprobassumption}, we assume that $\Gamma=\gamma I$ and that  $A$ has eigenvalues $\Big\{\frac{j^{-\beta}} {\gamma}\Big\}_{j=1}^\infty$ with $\gamma>0,$ and $ \beta \ge 0.$ We consider a truncated sequence of problems   with $A(\beta,\gamma,d),$  with eigenvalues $\Big\{\frac{j^{-\beta}} {\gamma}\Big\}_{j=1}^d$, $d\in  \N\cup \{\infty\}$.
Finally, we assume that the data is generated from a fixed underlying infinite dimensional truth $u^\dagger,$
\begin{equation}\label{eq:trdata}
y=\fp u^\dagger+\eta, \quad \fp u^\dagger\in\h,
\end{equation}
and for the truncated problems the data is given by projecting $y$ onto the first $d$ eigenfunctions of $A.$
\end{assumption}

\begin{remark}
Since $\fp u^\dagger\in\h$, using the Gaussian theory provided in the Supplementary Material one can check that the distribution of the data in \eqref{eq:trdata} is equivalent
to the marginal probability measure of the data under the model, $\mar(dy)$.  Hence, the conclusions of Theorem \ref{maintheorem} and Remark \ref{rem:main}  which are formulated for $\mar$-almost all $y$, also hold for almost all $y$ of the form of \eqref{eq:trdata}.
\end{remark}

Note that $d$ in the previous assumption is the data space dimension, which agrees here with the nominal dimension.
The setting of the previous assumption arises, for example, when $d$ is finite, from discretizing the data of 
an inverse problem formulated in an infinite dimensional state space. Provided that the forward map $K$ and the prior covariance $\Sigma$ commute, our analysis extends to the case where both the unknown and the data are discretized in the common eigenbasis.  In all these cases,  interest lies in understanding
how the cost of importance sampling depends on the level
of the discretizations. The parameter $\gamma$ may arise as an
observational noise scaling, and it is hence of interest to study 
the cost of importance sampling when $\gamma$ is small. And finally
the parameter $\beta$ reflects regularity of the problem, as determined
by the prior and noise covariances, and the forward map; critical phase
transitions occur in computational cost as this parameter is
varied, as we will show.

\begin{example}[Example \ref{ex:lm1} revisited]\label{ex:lm11}
We revisit the deconvolution problem in the unit interval. In particular, we consider the problem of deconvolution of a periodic signal which is blurred by a periodic kernel and polluted by Gaussian white noise $N(0,\gamma I)$. This problem is diagonalized by the Discrete Fourier Transform, giving rise to a countable number of decoupled equations in frequency space of the form \[y_j=K_ju_j+\eta_j, \quad j \in\N.\] Here $u_j$ are the Fourier coefficients of the unknown signal $u$, $K_j$ the Fourier coefficients of the blurring kernel $\phi$ which is assumed to be known, and $\eta_j\stackrel{iid}{\sim} N(0,\gamma)$ the Fourier coefficients of the observational noise $\eta$.  Consider the case in which $K_j\asymp j^{-t}, \;t\geq0$; the case $t=0$ corresponds to the direct observation case while the bigger $t$ is the more severe the blurring. We put a Gaussian prior on $u$, $u\sim N(0,(-\Delta)^{-s}), \,s\geq0$, where $\Delta$ is the Laplacian with periodic boundary conditions on $(0,1)$, so that by the Karhunen-Loeve expansion $u_j=\sqrt{\kappa_j}\zeta_j$, with $\kappa_j\asymp j^{-2s}$ and $\zeta_j\stackrel{iid}{\sim} N(0,1)$. The larger $s$ is the higher the regularity of draws from the prior. In this case the operator $A$ has eigenvalues $\left\{c\frac{j^{-2t-2s}}{\gamma}\right\}_{j=1}^\infty$, where $c$ is independent of $j, \gamma$. For this example the value of $\beta$ in Assumption \ref{ass1} is $\beta=2t+2s$ and large values of $\beta$ correspond to more severe blurring and/or higher regularity of the prior. A natural way of discretizing this problem is to truncate the infinite sequence of 1-dimensional problems to $d$ terms, resulting in truncation of the sequence of eigenvalues of $A$. The limit $\gamma\to0$ corresponds to vanishing noise in the observation of the blurred signal.\end{example}

The intrinsic dimensions $\tau=\tau(\beta,\gamma,d)$ and $\effd=\effd(\beta,\gamma,d)$ read
\begin{equation}\label{tauefd}
\tau= \frac{1}{\gamma} \sum_{j=1}^d j^{-\beta}, \quad \quad  \effd= \sum_{j=1}^d \frac{j^{-\beta}}{\gamma + j^{-\beta}}.
\end{equation}

Table \ref{table:efdtaurhoscalings} shows the scalings of the effective 
dimensions $\effd$ and $\tau$ with the model parameters. It also shows
how $\rho$ behaves under
these scalings and hence gives, by Theorem \ref{non-asymptotictheorem},
an indication of the number of particles required for accurate
importance sampling in a given regime. In all the scaling limits
where $\rho$ grows to infinity the posterior and prior are approaching
mutual singularity; we can then apply Theorem \ref{non-asymptotictheorem}
to get an indication of how importance sampling deteriorates 
in these limits.

Note that by Theorem \ref{maintheorem} we have
$\tau(\beta,\gamma,d)<\infty$ if and only if $\effd(\beta,\gamma,d)<\infty.$ It
is clear from \eqref{tauefd} that $\tau=\infty$ if and only if
$\{d=\infty, \beta\le 1\}.$ By Theorem \ref{maintheorem} again, this
implies, in particular, that absolute continuity is lost in the limit as $d\to\infty$ when $\beta \le 1$, and as $\beta\searrow 1$ when $d=\infty.$ 
Absolute continuity is also lost in the limit $\gamma \to 0$, in which the posterior is fully concentrated around the data (at least in those directions in which the data live).  In this limit we always have $\tau=\infty$,  
whereas $\effd<\infty$ in the case where $d<\infty$ and $\effd=\infty$ when
$d=\infty$. Note that in the limit $\gamma=0$ Assumption
\ref{Generalinvprobassumption} does not hold, which explains why
$\tau$ and $\effd$ are not finite simultaneously. Indeed, as was noted
before, $\effd$ is always bounded by the nominal dimension $d$
irrespective of the size $\gamma$ of the noise.  

 Some important remarks on Table \ref{table:efdtaurhoscalings} are:
\begin{itemize}
\item $\rho$ grows \emph{algebraically} in the small noise limit ($\gamma\to 0$) if the nominal dimension $d$ is finite. 
\item $\rho$ grows \emph{exponentially} in $\tau$ or $\effd$ as the nominal dimension grows $(d\to \infty)$ if $\beta<1$ , and as the prior becomes rougher ($\beta \searrow 1)$ if $d=\infty.$
\item $\rho$ grows \emph{factorially} in the small noise limit ($\gamma\to 0$) if $d=\infty$, and in the joint limit $\gamma=d^{-\alpha}, d\to \infty$. The exponent in the rates relates naturally to $\effd.$
\end{itemize}

The scalings of $\tau$ and $\effd$ can be readily deduced by comparing the sums defining $\tau$ and $\effd$ with integrals. The analysis of the sensitivity of  $\rho$ to the model parameters relies on an explicit expression for this quantity. Details are given in the Supplementary Material.

\begin{table}

\begin{tabular}{l | c |c|c|c}\label{tablecomparison}
Regime & Parameters & $\effd$ & $\tau$ & $\rho$ \\\hline
Small noise & $\gamma\to 0, \,d<\infty$ &  $d$  & $\gamma^{-1}$ & $\gamma^{-d/2}$\\
                   & $\gamma\to 0, \,d=\infty, \,\beta>1$ &   $\gamma^{-1/\beta}$   & $\gamma^{-1}$ & $\gamma^{-\frac{\epsilon\beta}{2}(\gamma^{-1/\beta-\epsilon}) }$ \\ \hline
Large $d$  & $d\to \infty$, $\beta<1$  & $d^{1-\beta}$ & $d^{1-\beta}$ & $\exp(d^{1-\beta})$ \\\hline
Small noise  &$\gamma=d^{-\alpha},\, d\to \infty, \beta>1, \alpha>\beta$ & $d$ & $d^\alpha$ & $d^{(\alpha-\beta) d}$\\
and large $d$  &$\gamma=d^{-\alpha},\, d\to \infty, \beta>1, \alpha<\beta$ & $d^{\alpha/\beta}$ & $d^\alpha$ & $d^{\epsilon d^{\alpha/\beta-\epsilon}}$\\
                     & $\gamma=d^{-\alpha}, \, d\to \infty, \beta<1, \alpha>\beta$ & $d$ & $d^{1+\alpha-\beta}$ & $d^{(\alpha-\beta) d}$  \\
                      & $\gamma=d^{-\alpha}, \, d\to \infty, \beta<1, \alpha<\beta$ & $d^{1+\alpha-\beta}$ & $d^{1+\alpha-\beta}$  & $d^{\epsilon d^{\alpha/\beta-\epsilon}}$  \\\hline
Regularity & $d=\infty, \beta \searrow 1$ & $\frac{1}{\beta -1}$ & $\frac{1}{\beta -1}$ & $\exp(\frac{1}{\beta -1})$\\
\end{tabular}
\caption{\label{table:efdtaurhoscalings} The third and fourth columns show the scaling of the intrinsic dimensions with model parameters for the spectra cascade example of subsection \ref{sssec:spectralcascade}. The fourth one gives a lower bound on the growth of $\rho,$ suggesting that the number of particles should be increased \emph{at least} as indicated by this column in terms of the model parameters. This lower bound holds for all realizations of the data $y$ when $\gamma\to 0$, and in probability for those regimes where $\gamma$ is fixed. $\epsilon$ can be chosen arbitrarily small.}
\end{table}

\subsection{Discussion and Connection to Literature}
\label{ssec:lit2}
\subsubsection{Examples and Hilbert Space Formulation of Inverse Problems}
Further examples of linear inverse problems in both finite and infinite dimensions include the Radon transform inversion used for X-ray imaging, the determination of the initial temperature from later measurements and the inversion of the Laplace transform. Many case studies and more elaborate nonlinear inverse problems can be found for example in \cite{KS05}, \cite{AS10} which adopt a Bayesian approach to their solution, and \cite{EHN96}, \cite{MS12} which adopt a classical approach. The periodic deconvolution problem considered in Example \ref{ex:lm11} is discussed for instance  in \cite[Section 5]{CT02}, where an example of a convolution operator with algebraically decaying spectrum is also provided. The Bayesian approach we undertake, in the example of linear regression (Example \ref{ex:lm}) becomes the Gaussian conjugate Bayesian analysis of linear regression models, as in \cite{LS72}. This paper also derives
formulae \eqref{eq:fdposc}, \eqref{eq:fdposm} for the
mean and covariance expressed via precisions in the finite 
dimensional setting.
For the infinite dimensional counterpart see \cite[Section 5]{ALS13}.
Formulae \eqref{eq:posm}, \eqref{eq:posc} in the infinite dimensional
setting are derived in \cite{AM84}, \cite{LPS89}; in the specific
case of inverting for the initial condition in the heat equation 
they were derived in \cite{JF70}. 
The Supplementary Material
has a discussion of Gaussian measures in Hilbert
spaces and contains further background references.

%As mentioned above, we tacitly assume that $\fp$ can be extended to act on elements in $\X$ and that the sum of $\fp u$ and $\eta$ makes sense in $\Y$. This assumption holds trivially if the three operators $\fp, \prc, \nc$ are simultaneously diagonalizable. It also holds in non-diagonal settings, in which it is possible to link the domains of powers of the three operators by appropriate embeddings; for some examples see \cite[Section 7]{ALS13}.

\subsubsection{The Operator $A$: Centrality and Assumptions}
The assumption that the spectrum of $A$ introduced in Assumption \ref{Generalinvprobassumption} consists of a countable number of eigenvalues, means that the operator $A$ can be thought of as an infinitely large diagonal matrix. It holds if $A$ is compact \cite[Theorem 3, Chapter 28]{PDL02}, but is in fact more general since it covers, for example, the non-compact case $A=I$. 

\begin{comment}
In the finite dimensional setting the assumption that $A$ is bounded holds 
automatically if the noise covariance is invertible.
The centrality of $S=\nc^{-\frac12}\fp\prc^\frac12$ may then be understood as follows.
Under the prior and noise models 
we may write $u=\Sigma^{\frac12}u_0$ 
and $\eta=\Gamma^{\frac12}\eta_0$ where $u_0$ and $\eta_0$ are
independent centred Gaussians with identity covariance operators
(white noises).
Under the assumption that $\Gamma$ is invertible we then find
that we may write \eqref{invprob}, for $y_0=\Gamma^{-\frac12}y$, as
\begin{equation}
\label{eq:sip}
y_0=Su_0+\eta_0.
\end{equation}
Thus all results may be derived for this inverse problem, and translated
back to the original setting. The role of $S$, and hence $A$, 
is thus clear in the finite dimensional setting. 
This intuition carries over to infinite dimensions. % and, 
%furthermore, the assumption that $A$ is bounded has a natural  interpretation in terms of the ill-posedness of the problem. 
%\marginpar{\modd{Sergios/Daniel: please check, reword etc. S: actually we could modify Lemma 5.5(iv) to say that $\tr(Q*Q)=\tr(QQ*)$ always since if one of the two sides is finite then so $Q$ or $Q^*$ HS and the equality holds, otherwise both are infinite.}}
\end{comment}
 
 We note here that the inverse problem
\begin{equation}
\label{eq:sip2}
y_0=w_0+\eta_0
\end{equation}
with $\eta_0$ a white noise and $w_0 \sim \G(0,SS^\ast)$ is equivalent
to \eqref{eq:sip}, but formulated in terms of unknown $w_0=Au_0$, rather
than unknown $u_0.$ In this picture the key operator is $SS^\ast$
rather than $A=S^\ast S$.
Note that by Lemma \ref{lem:tr} in the Supplementary Material
$\tr(S^\ast S)=\tr(SS^\ast)$. Furthermore, if $S$ is compact the operators $SS^\ast$ and $S^\ast S$ have the same nonzero eigenvalues \cite[Section 2.2]{EHN96}, thus $\tr((I+SS^\ast)^{-1}SS^\ast)=\tr((I+S^\ast S)^{-1}S^\ast S)$. The last equality holds even if $S$ is non-compact, since then Lemma \ref{lem:tr} together with Lemma \ref{propequivefd} imply that both sides are infinite. Combining, we see that the intrinsic dimension ($\tau$ or $\effd$) is the same regardless of whether we view $w_0$ or $u_0$ as the
unknown.
In particular, the assumption that $A$ is bounded is equivalent to 
assuming that the operators $S, S^\ast$ or  $SS^\ast$ are bounded \cite[Theorem 14, Chapter 19]{PDL02}. %The equivalent formulation \eqref{eq:sip2} of the linear inverse problem, one can that the condition that $A$ is bounded is the borderline case for having that the prior acts as a regularization in the inversion of the data. 
For the equivalent formulation \eqref{eq:sip2}, the posterior mean equation \eqref{eq:posm} is \[\posm=SS^\ast(SS^\ast+I)^{-1}y.\] If $SS^\ast$ is compact, that is, if its nonzero eigenvalues $\lambda_i$ go to $0$, then $m$ is a regularized approximation of $w_0$, since the components of the data corresponding to small eigenvalues $\lambda_i$ are shrunk towards zero. On the other hand, if $SS^\ast$ is unbounded, that is, if its nonzero eigenvalues $\lambda_i$ go to infinity, then there is no regularization and high frequency components in the data remain almost unaffected by $SS^\ast$ in $m$. Therefore, the case $SS^\ast$ is bounded is the borderline case determining whether the prior has a regularizing effect in the inversion of the data.  %The equivalent formulations \eqref{eq:sip} and \eqref{eq:sip2} of the linear inverse problem show that the condition that $A$ is bounded is the borderline case  for having an ill-posed problem under the given prior beliefs on the regularity of the unknown. \ds{Indeed, as soon as we assume that $A$ (equivalently $SS^\ast$) is compact, that is its eigenvalues $\lambda_i$ converge to $0$,  we have that under the prior the signal $w_0$ is smoother than the noise $\eta_0$. Hence we cannot recover stably the signal in the support of the prior, since the identity viewed as an operator mapping the support of the noise to the support of the prior is unbounded.}

 The operator $A$ has played an important role in the study of
linear inverse problems. First, it has been used for obtaining posterior contraction rates in the small noise limit, see the operator $B^\ast B$ in \cite{LLM15}, \cite{AM14}. Its use was motivated by techniques for analyzing classical regularization methods, in particular regularization in Hilbert scales see \cite[Chapter 8]{EHN96}. Furthermore, its eigenvalues and eigendirections can be used to determine (optimal) low-rank approximations of the posterior covariance \cite{BGMS13}, \cite[Theorem 2.3]{marzouki2014optimal}. The analogue of $A$ in nonlinear Bayesian inverse problems is the so-called prior-preconditioned data-misfit Hessian, which has been used in \cite{MWBG12} to design Metropolis Hastings proposals. In more realistic settings the spectrum of $A$ may not be analytically available and needs to be numerically approximated; for example see \cite[subsection 6.7]{BGMS13} in the context of linearized global seismic inversion.

\subsubsection{Notions of Dimension and Interpretations}
\label{ssec:stat-dim}
In subsection \ref{ssec:ess2} we study notions of dimension for Bayesian inverse problems.
In the Bayesian setting, the prior imparts information and correlations on the components of
the unknown $u$, reducing the 
number of parameters that are estimated. 
In the context of  Bayesian or penalized likelihood 
frameworks, this has led to the notion of \emph{effective 
number of parameters}, defined as
\[
\tr\Bigl(\Gamma^{1/2} \auxm ( I + \auxm^* \auxm)^{-1} \auxm^* \Gamma^{-1/2}\Bigr).
\]
This quantity agrees with $\effd$ by Proposition \ref{prop:efd}  and has been used extensively in statistics and machine
learning ---see for example the deviance information criterion in \cite{dic} (which generalises this notion to more general Bayesian hierarchical models), and section 3.5.3
of \cite{bishop} and references therein.
One motivation for this  definition is based on a
Bayesian version of the  ``hat matrix'', see for example
\cite{dic}. In this article we provide a different motivation
that is more relevant to our aims: rather than as an effective
number of parameters, we interpret $\effd$ as the effective
  dimension of the Bayesian linear model.
Similar forms of effective dimension have been used for learning problems in \cite{zhang2002effective}, \cite{zhang2005learning}, \cite{caponnetto2007optimal} and for statistical inverse problems in \cite{lu2014discrepancy}. In all of these contexts  the size of the operator $A$  quantifies how informative the data is; see the discussion below. The paper \cite{BLB08} introduced the notion of $\tau= \tr(A)$ as an
effective dimension for importance sampling within linear inverse problems
and filtering. In that paper several transformations of the inverse problem 
are performed before doing the analysis. We undo these
transformations here. The role of $\tau$ in the performance of the Ensemble Kalman filter had been previously studied in \cite{furrer2007estimation}.

\begin{comment}This eigenvalue problem has been widely studied in the context of classical regularization of linear inverse problems, e.g. \cite{hansen1989regularization}, \cite{dykes2014simplified}. 
\end{comment}

Proposition \ref{prop:efd} shows that $\effd$ is at most as large as the 
nominal dimension. The difference between both is a measure of the effect the prior has on the inference
relative to the maximum likelihood solution.
%This difference increases as the size of $\Sigma$ increases, or as the correlation among the vectors that form the columns
%of $K$ increases, while the difference decreases as the size of $\Gamma$ decreases
%or as the correlations in $\Gamma$
%increase.
Proposition \ref{prop:new}  shows that
$\effd$ quantifies how far the posterior is from the prior, measured in terms of how distant
their covariances are in units of the prior; and similarly for $\tau,$ but expressed in terms
of precisions and again in units of the prior. By the cyclic property of the trace, Lemma \ref{lem:tr}(ii) in the Supplementary Material, and by Proposition \ref{prop:new}, $\tau$ and $\effd$ may also be
characterized as follows: 
\begin{align*}
\tau &= \tr\bigl((C^{-1}-\Sigma^{-1})\Sigma\bigr)=\tr  \bigl( (\Sigma - C)C^{-1} \bigr),  \\
\effd &= \tr\bigl((\Sigma-C)\Sigma^{-1}\bigr)=\tr\bigl( (C^{-1} - \Sigma^{-1}) C \bigr).
 \end{align*}
Thus we may also view $\effd$ as measuring the change in
the precision, measured in units given by the posterior
precision; whilst $\tau$ measures the change in the covariance,
measured in units given by the posterior covariance.

\section{Importance Sampling and Filtering}
\label{sec:FIL}

This section studies importance sampling in the context of filtering. 
In particular we study two different choices of proposals that play 
an important role in the subject of filtering. The analysis relies 
on the relationship between Bayesian inversion and filtering mentioned
in the introductory section, and detailed here.
In subsection \ref{ssec:gen3} we set-up the problem and derive
a link between importance sampling based particle filters and
the inverse problem. In subsections \ref{ssec:ess3} and \ref{ssec:con3}
 we use this connection to study, respectively, the intrinsic
dimension of filtering and the connection to absolute
continuity between the two proposals considered and the target. 
Subsection \ref{ssec:sin3} contains some explicit computations which
enable comparison of the cost of the two proposals
in various singular limits relating to high dimension or
small observational noise.
We conclude with the literature review subsection \ref{ssec:lit3}.

\begin{comment}
In section \ref{sec:IS} we introduced importance sampling,
and studied its computational cost. We highlighted the
role of the density of the target with respect to the proposal. We also studied 
the behaviour of importance sampling when approaching 
 loss of absolute continuity between target and proposal.
In particular we studied the effect of various
singular limits (large nominal dimension, small parameters) in this
breakdown. Section \ref{sec:BIP} studied these issues 
 for Bayesian linear inverse problems. Here we study
them for the filtering problem,
using the relationship between Bayesian inversion and filtering
outlined in the introductory section, and detailed here.
In subsection \ref{ssec:gen3} we set-up the problem and derive
a link between importance sampling based particle filters and
the inverse problem. In subsections \ref{ssec:ess3} and \ref{ssec:con3}
respectively we use this connection to study the intrinsic
dimension of filtering, and the connection to absolute
continuity between proposal and target, and in doing so
make comparisons between the standard and optimal proposals.
Subsection \ref{ssec:sin3} contains some explicit computations which
enable comparison of the cost of the two proposals
in various singular limits relating to high dimension or
small observational noise.
We conclude with the literature review subsection \ref{ssec:lit3}
which overviews the sources for the material herein. 
\end{comment}

The component of particle filtering which we analyse in this section 
is only that related to sequential importance sampling; we do not
discuss the interaction between the simulated particles which arises
via resampling schemes. Such interaction would not typically be 
very relevant in the two time-unit dynamical systems we study here, 
but would be necessary to get reasonable numerical schemes when assimilating 
data over many time units. We comment further on this, and the choice 
of the assimilation problem we study, in the literature review.

\subsection{General Setting}
\label{ssec:gen3}

We simplify the notation by setting $j=0$ in
\eqref{Filteringproblem} to obtain 
\begin{align}\label{eq:shift}
\begin{split}
v_{1} &= \signalmap v_0 + \xi,  \quad v_0\sim \G(0,P),\quad \xi\sim \G(0,\signalc),\\
y_{1} &= \obsop v_{1} + \zeta, \quad \zeta \sim  \G(0,\obsc).
\end{split}
\end{align} 
Note that we have also imposed a Gaussian assumption on $v_0.$
Because of the Markov assumption on the dynamics for $\{v_j\}$, 
we have that $v_0$ and $\xi$ are independent.
As in section \ref{sec:BIP} we set-up the problem in a
separable Hilbert space $\h$, although the reader versed only
in finite dimensional Gaussian measures should have no trouble
following the developments, simply by thinking of the
covariance operators as (possibly infinite) matrices. We assume
throughout that the covariance operators $\zeroc, \signalc, \obsc:\h \to \h$ are bounded, self-adjoint, positive linear operators, but not necessarily trace-class (see the discussion on this trace-class issue 
in section \ref{sec:BIP}).
We also assume that the operators $\signalmap, \obsop: \h\to\h$ that describe, respectively, the unconditioned signal dynamics and the observation operator, can be extended to larger spaces if necessary; see 
the Supplementary Material for
further details on these technical issues.

Our goal in this section is to study the cost of
importance sampling within the context of both the standard
and optimal proposals for particle filtering. For both these
proposals we show that there is an inverse problem embedded
within the particle filtering method, and compute the proposal covariance,
the observation operator and the observational noise covariance.
We may then use the material from the previous section, concerning
inverse problems, to make direct conclusions about the
cost of importance sampling for particle filters.

The aim of one step of filtering may be expressed as sampling
from the target $\law_{v_1,v_0|y_1}.$ Particle filters do this
by importance sampling, with this measure on the product space
$\X \times \X$ as the target. We wish to 
compare two ways of doing this, one by using the proposal distribution 
$\law_{v_1|v_0}\law_{v_0}$ and the second by using as proposal
distribution $\law_{v_1|v_0,y_1}\law_{v_0}.$ The first is
known as the {\em standard proposal}, and the second as the 
{\em optimal proposal.} We now connect each of these proposals 
to a different inverse problem. 

\subsubsection{Standard Proposal}

For the standard proposal we note that, using Bayes' theorem, conditioning,
and that the observation $y_1$ does not depend on $v_0$ explicitly,
\begin{align*}
\law_{v_1,v_0|y_1}&\propto\law_{y_1|v_1,v_0}\law_{v_1,v_0}\\
&=\law_{y_1|v_1,v_0}\law_{v_1|v_0}\law_{v_0}\\
&=\law_{y_1|v_1}\law_{v_1|v_0}\law_{v_0}.
\end{align*}
Thus the density of the target $\law_{v_1,v_0|y_1}$ 
with respect to the proposal $\law_{v_1|v_0}\law_{v_0}$ is
proportional to $\law_{y_1|v_1}$. Although this density concerns
a proposal on the joint space of $(v_0,v_1)$, since it involves
only $v_1$ we may consider the related inverse problem of finding
$v_1$, given $y_1$, and ignore $v_0.$

In this picture filtering via the standard proposal 
proceeds as follows:
$$\law_{v_0} \mapsto \law_{v_1} \mapsto \law_{v_1|y_1}.$$
Here the first step involves propagation of proability measures
under the dynamics. This provides the proposal $\prm=\law_{v_1}$
used for importance sampling to determine the target $\tm=\law_{v_1|y_1}.$
The situation is illustrated in the upper branch of Figure \ref{twopaths}.
Since 
$$\E (v_1 v_1^*)=\E (Mv_0+\xi)(Mv_0+\xi)^*,$$
and $v_0$ and $\xi$ are independent under the Markov assumption,
the proposal distribution is readily seen to be a centred Gaussian
with covariance $\Sigma=MPM^*+Q$. The observation operator is $K=H$
and the noise covariance $\Gamma=R.$ 
We have established a direct connection between the particle
filter, with standard proposal, and the inverse problem of
the previous section. We will use this connection to study the
cost of the particle filter, with standard proposal, in what
follows.

\subsubsection{Optimal Proposal}

For the optimal proposal we note that, by conditioning on $v_0$,
\begin{align*}
\law_{v_1,v_0|y_1}&=\law_{v_1|v_0,y_1}\law_{v_0|y_1}\\
&=\law_{v_1|v_0,y_1}\law_{v_0}\frac{\law_{v_0|y_1}}{\law_{v_0}}.
\end{align*}
Thus the density of the target $\law_{v_1,v_0|y_1}$
with respect to the proposal $\law_{v_1|v_0,y_1}\law_{v_0}$ is
the same as the density of $\law_{v_0|y_1}$ with respect to
$\law_{v_0}$. As a consequence, although this density concerns
a proposal on the joint space of $(v_0,v_1)$, it is
equivalent to an inverse problem involving only $v_0.$  
We may thus consider the related inverse problem of finding
$v_0$ given $y_1$, and ignore $v_1.$

In this picture filtering via the optimal proposal proceeds as follows:
$$\law_{v_0} \mapsto \law_{v_0|y_1} \mapsto \law_{v_1|y_1}.$$
Here the first step involves importance sampling with
proposal $\prm=\law_{v_0}$ and target $\tm=\law_{v_0|y_1}.$ 
This target measure is then propagated under the conditioned dynamics to
find $\law_{v_1|y_1};$ the underlying assumption of the optimal
proposal is that $\law_{v_1|v_0,y_1}$ can be sampled so that this
conditioned dynamics can be implemented particle by particle.
The situation is illustrated in the lower branch of Figure \ref{twopaths}.
Since
$$y_1=HMv_0+H\xi+\zeta$$
the proposal distribution is readily seen to be a centred Gaussian
with covariance $\Sigma=P$, the observation operator $K=HM$
and the noise covariance given by the covariance of $H\xi+\zeta$,
namely $\Gamma=HQH^*+R.$
Again we have established a direct connection between the particle 
filter, with optimal proposal, and the inverse problem of the previous 
section. We will use this connection to study the cost of the 
particle filter, with optimal proposal, in what follows.

A key assumption of the optimal proposal is the second step: the
ability to sample from the conditioned dynamics $\law_{v_1|v_0,y_1}$ and
we make a few comments on this before returning to our main purpose,
namely to study cost of particle filtering via the connection
to an inverse problem. The first comment is
to note that since we are in a purely Gaussian setting, this
conditioned dynamics is itself determined by a Gaussian and so
may in principle be performed in a straightforward fashion. In
fact the conditioned dynamics remains Gaussian even if the
forward model $Mv_0$ is replaced by a nonlinear map $f(v_0)$,
so that the optimal proposal has wider applicability than might
at first be appreciated. Secondly we comment that the Gaussian
arising in the conditioned dynamics
has mean $m$ and variance $\Xi$ given by the formulae
\begin{align*}
\Xi&=Q-QH^*(HQH^*+R)^{-1}HQ,\\
m&=Mv_0+QH^*(HQH^*+R)^{-1}(y_1-HMv_0).
\end{align*} 
It is a tacit assumption in what follows that the operators defining
the filtering problem are such that $\Xi: \h \to \h$ is well-defined
and that $m \in \h$ is well-defined.
More can be said about these points, but doing so will add further
technicalities without contributing to the main goals of this
paper.

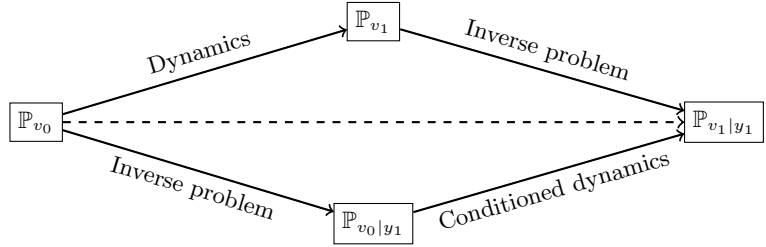
\begin{figure}
\begin{center}
\begin{tikzpicture}
\matrix [column sep=36mm, row sep=8mm] {
 &
  \node (yw) [draw, shape=rectangle] {$\law_{v_1}$}; &
   \\
  \node (d1) [draw, shape=rectangle] {$\law_{v_0}$}; & &
  \node (d2) [draw, shape=rectangle] {$\law_{v_1|y_1}$}; \\
  &
  \node (ec) [draw, shape=rectangle] {$\law_{v_0|y_{1}}$}; &
  \\
};
\draw[->, thick,dashed] (d1) -- (d2);
\draw[->, thick] (d1) -- (yw);
\draw (d1) -- (yw) node [midway,above,sloped] (TextNode) {Dynamics};
\draw[->, thick] (d1) -- (ec);
\draw (d1) -- (ec) node [midway,below,sloped] (TextNode) {Inverse problem};
\draw[->, thick] (yw) -- (d2);
\draw (yw) -- (d2) node [midway,above,sloped] (TextNode) {Inverse problem};
\draw[->, thick] (ec) -- (d2);
\draw (ec) -- (d2) node [midway,below,sloped] (TextNode) {Conditioned dynamics};
\end{tikzpicture}
\end{center}
\caption{\label{twopaths} Filtering step decomposed in two different ways. The upper path first pushes forward the measure $\law_{v_0}$ using the signal dynamics, and then incorporates the observation $y_{1}.$ The lower path assimilates the observation $y_{1}$ first, and then propagates the conditioned measure using the signal dynamics. The standard proposal corresponds to the upper decomposition and the optimal one to the lower decomposition.}
\end{figure}

\subsection{Intrinsic Dimension}
\label{ssec:ess3}

Using the inverse problems that arise for the standard proposal
and for the optimal proposal, and employing them within the definition of
$A$ from Assumption \ref{Generalinvprobassumption},
we find the two operators $A$ arising for these two different proposals:
$$A:=A_{st}:=(\signalmap \zeroc \signalmap^* + \signalc)^{1/2}\obsop^* \obsc^{-1} \obsop (\signalmap \zeroc \signalmap^* + \signalc)^{1/2}$$
for the standard proposal, and 
$$A:=A_{op}:=\zeroc^{\frac12} \signalmap^* \obsop^* (\obsc + \obsop \signalc \obsop^*)^{-1} \obsop \signalmap \zeroc ^{1/2}$$
for the optimal proposal.
Again here it is assumed that these operators are bounded in $\h$:

\begin{assumption}
\label{a:filt}
The operators $A_{st}$ and $A_{op}$, viewed as linear operators in $\h$,
are bounded. Furthermore, the spectra of both $A_{st}$ and $A_{op}$ consist of a countable number of eigenvalues.
\end{assumption}

Using these definitions of $A_{st}$ and $A_{op}$ we may define,
from \eqref{eq:efd}, the intrinsic dimensions $\tau_{st},\,\effd_{st}$
for the standard proposal and $\tau_{op},\,\effd_{op}$ for the optimal one
in the following way
\begin{align*}
\tau_{st} = \tr (A_{st}), \quad \quad  \effd_{st} = \tr\bigl((I+A_{st})^{-1}A_{st}\bigr)
\end{align*}
and
\begin{align*}
\tau_{op} = \tr (A_{op}), \quad \quad  \effd_{op} = \tr\bigl((I+A_{op})^{-1}A_{op}\bigr).
\end{align*}

\begin{table}
\begin{tabular}{l | c |c}
& Standard Proposal & Optimal proposal  \\\hline
Proposal & $\law_{v_0}(dv_0)\law_{v_1|v_0}(dv_{1})$ & $\law_{v_0}(dv_0)\law_{v_1|v_0,y_1}(dv_1)$\\
BIP  & $y_1=\obsop v_1 + \eta_{st}$ & $y_1=\obsop \signalmap v_0 +\eta_{op}$ \\
Prior Cov.  & $\signalmap \zeroc \signalmap^* + \signalc$ & $\zeroc$  \\
Data Cov. & $\obsc$ & $\obsc + \obsop \signalc \obsop^*$\\
$\log g(u;y_1)$ & $ -\frac{1}{2}\|\obsop v_1\|_{\obsc}^2 + \langle y_1,  \obsop v_1 \rangle_{\obsc}$&  $ -\frac{1}{2}\|\obsop \signalmap v_0\|_{\obsc + \obsop \signalc \obsop^*}^2 + \langle y_1,  HMv_0\rangle_{R +HQH^*}$\\
\end{tabular}
\caption{\label{table:filteringequiv}}
\end{table}

\begin{comment}
Please say more useful stuff about the $\tau$'s and $\effd$'s
here. In fact define them all and then mention all the cyclic properties
that are used in what follows with regard to trace calculations. The equivalent expressions have been moved down.
\end{comment}

\subsection{Absolute Continuity}
\label{ssec:con3}

The following two theorems are a straightforward application
of Theorem \ref{maintheorem}, using the connections between
filtering and inverse problems made above. The contents of the
two theorems are summarized in Table \ref{table:filteringequiv}.

\begin{theorem}\label{definitionrelatedinverseproblem}
Consider one-step of particle filtering for \eqref{eq:shift}. 
Let $\tm=\rp_{v_1|y_1}$ and $\prm=\rp_{v_1}=\G(0,Q+MPM^*).$ 
Then the following are equivalent:
\begin{enumerate}[i)]
\item $\effd_{st}<\infty;$
\item $\tau_{st}<\infty;$
\item $R^{-1/2}Hv_1\in \h$, $\prm$-almost surely;
\item for $\mar$-almost all $y$, the target distribution $\tm$ is well defined as a measure in $\X$  and is absolutely continuous with respect to the proposal with 
\begin{equation}\label{rninfdim1}
\frac{d\tm}{d\propo}(v_1)\propto\exp\left(-\frac{1}{2}\norm{R^{-1/2}H v_1}^2+\frac12\pr{R^{-1/2}y_1}{R^{-1/2}H v_1}\right)=:\rnip_{st}(v_1;y_1),
\end{equation} 
where $0<\prm\bigl(\rnip_{st}(\cdot;y_1)\bigr)<\infty.$
\end{enumerate}
\end{theorem}

\begin{theorem}\label{t:2}
Consider one-step of particle filtering for \eqref{eq:shift}. 
Let $\tm=\rp_{v_0|y_1}$ and $\prm=\rp_{v_0}=\G(0,Q).$ 
Then, for $R_{op}=R+HQH^*$, the following are equivalent:
\begin{enumerate}[i)]
\item $\effd_{op}<\infty;$
\item $\tau_{op}<\infty;$
\item $R_{op}^{-1/2}HMv_0\in \h$, $\prm$-almost surely;
\item for $\mar$-almost all $y$, the target distribution $\tm$ is well defined as a measure in $\X$  and is absolutely continuous with respect to the proposal with 
\begin{equation}\label{rninfdim3}
\frac{d\tm}{d\propo}(v_0)\propto\exp\left(-\frac{1}{2}\norm{R_{op}^{-1/2}HM v_0}^2+\frac12\pr{R_{op}^{-1/2}y_1}{R_{op}^{-1/2}HM v_0}\right)=:\rnip_{op}(v_0;y_1),
\end{equation} 
where $0<\prm\bigl(\rnip_{op}(\cdot;y_1)\bigr)<\infty.$
\end{enumerate}
\end{theorem}

\begin{remark}
\label{rem:fil}
Because of the exponential structure of $g_{st}$ and $g_{op}$, 
the assertion $(iv)$ in the preceding two theorems is 
equivalent to $g_{st}$ and $g_{op}$ being $\nu$-almost surely positive and finite and for almost all $y_1$ the second moment of the
target-proposal density being finite. This second moment is given, for
the standard and optimal proposals, by
   $$\rho_{st}=\frac{\propo\left(\rnip_{st}(\cdot;y)^2\right)}{\propo\bigl(\rnip_{st}(\cdot;y)\bigr)^2}<\infty$$
and
   $$\rho_{op}=\frac{\propo\left(\rnip_{op}(\cdot;y)^2\right)}{\propo\bigl(\rnip_{op}(\cdot;y)\bigr)^2}<\infty$$
respectively. The relative sizes of $\rho_{st}$ and $\rho_{op}$ determine
the relative efficiency of the standard and optimal proposal
versions of filtering.
\end{remark}

The following theorem shows that there is loss of absolute continuity for the standard proposal whenever there is for the optimal one.
The result is formulated in terms of the intrinsic dimension $\tau,$ and we show that $\tau_{op}=\infty$ implies $\tau_{st}=\infty;$ by Theorem \ref{maintheorem}, this implies the result concerning absolute continuity. 
 Recalling that
poor behaviour of importance sampling is intimately related
to such breakdown, this suggests that the optimal proposal is always
at least as good as the standard one.
The following theorem also gives a condition on the operators $H,\, Q$ and $R$ under which collapse for both proposals occurs at the same time, irrespective of the regularity of the operators $M$ and $P.$ Roughly
speaking this simultaneous collapse result states that if $R$ is large 
compared to $Q$ then absolute continuity for both proposals 
is equivalent; and hence collapse of importance sampling happens
under one proposal if and only if it happens under the other.
Intuitively the advantages of the optimal proposal stem from the
noise in the dynamics; they disappear completely if the dynamics is
deterministic. The theorem quantifies this idea.
Finally, an example demonstrates that there are situations where
$\tau_{op}$ is finite, so that optimal proposal based
importance sampling works well for finite dimensional 
approximations of an infinite dimensional problem, 
whilst $\tau_{st}$ is infinite,
so that standard proposal based
importance sampling works poorly for finite dimensional 
approximations.

\begin{theorem}\label{theoremcomparisoncollapse}
Suppose that Assumption \ref{a:filt}  holds.
Then,
\begin{equation}\label{TAUBD}
\tau_{op}\le \tau_{st}.
\end{equation}
Moreover, if $\tr\Bigl(\obsop \signalc \obsop^* \obsc^{-1}\Bigr)<\infty,$ then 
 $$\tau_{st}<\infty \iff \tau_{op}<\infty.$$
\end{theorem}

We remark that, under additional simplifying assumptions, we can obtain bounds of the form $\eqref{TAUBD}$ for $\effd$ and $\rho.$ We chose to formulate the result in terms of $\tau$ since we can prove the bound $\eqref{TAUBD}$ in full generality. Moreover,  by Theorem \ref{maintheorem} the bound in terms of $\tau$ suffices in order to understand the different collapse properties of both proposals.

The following example demonstrates that  it is possible
that $\tau_{op}<\infty$ while $\tau_{st}=\infty;$ in this
situation filtering via the optimal proposal is well-defined,
whilst using the standard proposal it is not.
Loosely speaking, this happens if $y_1$ provides more 
information on $v_1$ than $v_0.$

\begin{example}\label{EX}
Suppose that 
$$\obsop=\signalc=\obsc=\signalmap=I,\quad \quad \tr(P)<\infty.$$
Then, it is straightforward from the definitions that $A_{st}= P+I$ and $A_{op} = P/2.$ In an infinite dimensional Hilbert space setting the identity operator has infinite trace, $\tr(I) = \infty$, and so 
$$\tau_{st} =\tr(A_{st})=\tr(P+I) = \infty, \quad \quad \tau_{op}= \tr(A_{op})=\tr(P/2)<\infty.$$  We have thus established an example of a filtering model for which $\tau_{st}=\infty$ and $\tau_{op}<\infty.$
We note that by Theorem \ref{theoremcomparisoncollapse}, any such example satisfies the condition $\tr(HQH^*R^{-1}) = \infty.$  When this condition is met, automatically $\tau_{st}=\infty$. However, $\tau_{op}$ can still be finite. Indeed, within the proof of that theorem in the Supplementary Material we show that the
inequality 
 $$\tau_{op}\le \tr(R^{-1}HMPM^*H^*)$$ 
always holds.
The right-hand side may be finite provided that the eigenvalues of $P$ decay fast enough. A simple example of this situation is where $HM$ is a bounded operator and all the relevant operators have eigenvalues. In this case the Rayleigh-Courant-Fisher theorem --see the Supplementary Material-- guarantees that the eigenvalues of $HMPM^*H^*$ can be bounded in terms of those of $P$. Again by the Rayleigh-Courant-Fisher theorem, since we are always assuming that the covariance $R$ is bounded, it is possible to bound the eigenvalues of $R^{-1}HMPM^*H^*$ in terms of those of $HMPM^*H^*.$  This provides a wider range of examples where $\tau_{st}=\infty$ while $\tau_{op}<\infty.$
\end{example}

\subsection{Large Nominal Dimension and Singular Parameter Limits}
\label{ssec:sin3}

We noted in Remark \ref{rem:fil} that the values of the 
second moment of the target-proposal
density, $\rho_{st}$ and 
$\rho_{op}$, can be used to characterize the performance of particle filters based on the standard 
and optimal proposals, respectively. By comparing  
the values of $\rho_{st}$ and $\rho_{op}$ we can ascertain
situations in which the optimal proposal has significant
advantage over the standard proposal. We also recall, from section
\ref{sec:BIP}, the role of the intrinsic dimensions in determining
the scaling of the second moment of the target-proposal density.

The following example will illustrate a number of interesting
phenomena in this regard. In the setting of fixed finite state/data state dimension
it will illustrate how the scalings of the various covariances entering
the problem effect computational cost. 
In the setting of increasing nominal dimension $d$, when the limiting target
is singular with respect to the proposal, it will illustrate how
computational cost scales with $d$. And finally
we will contrast the cost of the filters in two differing
initialization scenarios: (i) from an arbitrary initial covariance $P$,
and from a steady state covariance $P_{\infty}.$
Such a steady state covariance is a fixed point of the covariance update
map for the Kalman filter defined by \eqref{Filteringproblem}.

\begin{table}
\begin{tabular}{l | c |c|c|c|c|c}
Regime & Param.& $\text{eig}(A_{st})$ & $\text{eig}(A_{op})$ & $\text{eig}(P_\infty)$ & $\rho_{st}$  &$\rho_{op}$ \\\hline
Small  obs. noise & $r\to 0$ &  $r^{-1}$  & $r$ &$r$ & $r^{-d/2}$& $1$\\
                   & $r=q\to 0$ &   $1$   & $1$ & $r(=q)$ &$1$ &$1$\\ \hline
Large $d$  & $d\to \infty$  & $1$ & $1$ &1& $\exp(d)$ &$\exp(d)$\\
\end{tabular}
\caption{\label{tablefilteringstationarity} Scalings of the standard and optimal proposals in small noise and large $d$ regimes for one filter step initialized from stationarity $(P=P_\infty$). This table should be interpreted 
in the same way as Table  \ref{table:efdtaurhoscalings}.}
\end{table}

\begin{table}
\begin{tabular}{l | c |c|c|c|c}
Regime & Param.& $\text{eig}(A_{st})$ & $\text{eig}(A_{op})$ & $\rho_{st}$  &$\rho_{op}$ \\\hline
Small obs. noise & $r\to 0$ &  $r^{-1}$  & $1$ & $r^{-d/2}$& $1$\\
                    & $r=q\to 0$ &   $r^{-1}$   & $r^{-1}$ &$r^{-d/2}$ &$r^{-d/2}$\\ \hline
Large $d$  & $d\to \infty$  & $1$ & $1$ & $\exp(d)$ &$\exp(d)$\\
\end{tabular}
\caption{\label{tablefilterinoutsidestationarity} Scalings of the standard and optimal proposals in small noise and large $d$ regimes for one filter step initialized from $P=pI $.  This table should be interpreted 
in the same way as Table  \ref{table:efdtaurhoscalings}.}
\end{table}

\begin{example}
Suppose that $\signalmap=\obsop= I\in \R^{d\times d}$, and $R=rI,$ $Q=qI$, with $r, q>0.$ A simple calculation shows that the steady state covariance is given by
\begin{equation*}
P_{\infty} = \frac{\sqrt{q^2 +4qr}-q}{2} \, I,
\end{equation*}
and that the operators $A_{st}$ and $A_{op}$ when $P=P_\infty$ are
\begin{equation*}
A_{st} = \frac{\sqrt{q^2 + 4qr}+q}{2r} \,I, \quad \quad A_{op} = \frac{\sqrt{q^2 + 4qr}-q}{2(q+r)}\, I.
\end{equation*}
Note that $A_{st}$ and $A_{op}$ are a function of $q/r$ only, whereas $P_\infty$ is not.

If the filtering step is initialized outside stationarity at $P=pI,$ with $p>0,$ then
\begin{equation*}
A_{st} =\frac{p+q}{r}\,I, \quad \quad A_{op} = \frac{p}{q+r}\, I.
\end{equation*}
Both the size and number of the eigenvalues of $A_{op}/A_{st}$ play
a role in determining the size of $\rho$, the second moment of the
target-proposal variance.  It is thus interesting to study
how $\rho$ scales in both the small observational noise regime $r \ll 1$
and the high dimensional regime $d \gg 1$. 
The results are summarized in 
Tables \ref{tablefilteringstationarity} and 
\ref{tablefilterinoutsidestationarity}. 
Some conclusions from these tables are:
\begin{itemize}
\item The standard proposal degenerates at an algebraic rate as $r\to 0$,
for fixed dimenson $d$, for both initializations of $P$. 
\item The optimal proposal is not sensitive to the small observation limit $r \to 0$ if the size of the signal noise, $q,$ is fixed. If started 
outside stationarity, the optimal proposal degenerates algebraically if 
$q \propto r\to 0$. However, even in this situation the optimal proposal 
scales well if initialized in the stationary regime.
\item In this example the limiting problem with $d=\infty$ has infinite 
intrinsic dimension for both proposals, because the target and the proposal
are mutually singular. As a result, $\rho$ grows exponentially in 
the large $d$ limit. 
\item Example \ref{EX} suggests that there are cases where $\rho_{st}$ grows exponentially in the large dimensional limit $d \to \infty$ but $\rho_{op}$ converges to a finite value. This may happen if $\tr\Bigl(\obsop \signalc \obsop^* \obsc^{-1}\Bigr)<\infty,$ but the prior covariance $P$ is sufficiently smooth.
\end{itemize}
\end{example}

\subsection{Discussion and Connection to Literature}
\label{ssec:lit3}

In subsection \ref{ssec:gen3}
we follow  \cite{BBL08}, \cite{BLB08}, \cite{snyder2008obstacles}, \cite{snyder2011alone}, \cite{slivinskipractical}, \cite{snyderberngtsson} and consider one step of the filtering model \eqref{Filteringproblem}. There are two main motivations for studying one step of the filter. Firstly,  if keeping the filter error small is prohibitively costly for one step, then there is no hope that an online particle filter will be successful \cite{BBL08}. Secondly,  it can provide insight for filters initialized close to stationarity \cite{chorin2013conditions}. As in  \cite{snyder2011alone}, \cite{slivinskipractical}, \cite{snyderberngtsson} we cast the analysis of importance sampling in joint space and consider as target $\tm:=\law_{u|y_1}$, with $u:=(v_0,v_1)$ and with the standard and optimal proposals defined in subsection \ref{ssec:gen3}.

In general nonlinear, non-Gaussian problems the optimal proposal is usually not implementable, since it is not possible to evaluate the corresponding weights, or to sample from the distribution $\law_{v_1|v_0,y_1}.$ 
However, the optimal proposal is implementable in our framework (see
for example \cite{doucet2000sequential}) and understanding its
behaviour is important in order to build and analyse improved and
computable proposals which are informed by the data
\cite{tu2013implicit}, \cite{goodman2014small},
\cite{van2010nonlinear}. It is worth making the point that the
so-called ``optimal proposal'' is really only \emph{locally} optimal. In
particular, this choice is optimal in minimizing the variance of the
weights at the given step given that all previous proposals
have been already chosen. This choice  does not minimize the Monte Carlo
variance for some time horizon for some family of test functions. A different optimality
criterion is obtained by trying to simultaneously minimize the variance of 
weights at times $t \leq r \leq t+m$, for some $m\geq 1$, or minimize
some function of these variances, say their sum or their maximum. Such
look ahead procedures might not be feasible in practice. Surprisingly,
examples exist where the standard proposal leads to smaller variance of
weights some steps ahead relative to the locally optimally tuned
particle filter;  see for example section 3 of
\cite{johansen2008note}, and the discussion in \cite[Chapter
10]{omirosb}.  Still, such examples are quite contrived and
experience suggests that local adaptation is useful in practice.

Similarly as for inverse problems, the values of $\rho_{st}$ and $\rho_{op}$ determine the performance of importance sampling for the filtering model with the standard and optimal proposals. In subsection \ref{ssec:con3} we show that the conditions of collapse for the  standard and optimal proposals (found in \cite{snyder2011alone} and \cite{BLB08}, respectively) correspond to any of the equivalent conditions of finite  intrinsic dimension or finite $\rho$  in Theorems \ref{definitionrelatedinverseproblem} and \ref{t:2}.

In subsection \ref{ssec:sin3} we study singular limits in the framework of \cite{chorin2013conditions}. Thus, we consider a  diagonal filtering setting in the Euclidean space $\R^d$, and assume that all coordinates of the problem play the same role, which corresponds to the extreme case $\beta=0$ in subsection \ref{ssec:sin2}. The paper \cite{chorin2013conditions} introduced a notion of effective dimension for detectable and stabilizable linear Gaussian data assimilation problems as the Frobenius norm of the steady state covariance of the filtering distribution. It is well known that the detectability and stabilizability conditions ensure the existence of such steady state covariance \cite{lancaster1995algebraic}.
This notion of dimension quantifies the success of data assimilation in having reduced uncertainty on the unknown once the data has been assimilated. Therefore the definition of dimension given in
\cite{chorin2013conditions} is at odds with both $\tau$ and $\effd:$ 
it does not quantify how much is learned from the data in one step, 
but instead how concentrated the filtering distribution is in the
time asymptotic regime when the filter is in steady state. 
Our calculations demonstrate differences which can occur
in the computational cost of filtering, depending on
whether it is initialized in this statistical steady state,
or at an arbitrary point. 
The  paper \cite{chorin2013conditions}  also highlights the importance of the size of the operator 
$A$ in studying the performance of importance sampling, both for the
standard and optimal proposals. Motivated by computational and physical 
intuition, the authors of \cite{chorin2013conditions} quantify the size of this operator by means of the Frobenius 
norm rather than the trace which we employ here. The trace is
more natural in the infinite dimensional limit, as demonstrated through 
large intrinsic dimension limits in \cite{BLB08}, and through the connection with 
absolute continuity in Theorems \ref{definitionrelatedinverseproblem} and \ref{t:2} above. We remark that the analysis 
in \cite{snyder2011alone} also relies on traces, but an unfortunate typo may trick the reader 
into believing that the Frobenius norm is being used. Note also
that some authors \cite{BLB08} write the eigenvalues as squares which can cause
 confusion to the casual reader.

\section{Conclusions}
In this article we have provided a framework which unifies 
the multitude of publications with bearing on importance sampling. 
We have aimed to give new insight into the potential use of importance sampling for inference in inverse problems
and filtering settings in models that
involve high and infinite state space and data dimensions.  
Our study has required revisiting the fundamental structure of importance
sampling on general state spaces. We have derived non-asymptotic
concentration inequalities for the particle approximation error and
related what turns out to be the key parameter of performance, the second moment of
the density between the target and proposal, to many different
importance sampling input and output quantities. 

As a  compromise between mathematical tractability and
practical relevance
we have focused on Bayesian linear models for regression and
statistical inversion of ill-posed inverse problems. We have studied
the efficiency of sampling-based posterior inference in these contexts
carried out by importance sampling using the prior as proposal. We
have demonstrated that performance is controlled by an intrinsic
dimension, as opposed to the state space or data dimensions, and we
have discussed and related two different notions of intrinsic 
dimension. It is important to emphasise that the intrinsic dimension
quantifies the relative strength between the prior and the
likelihood in forming the posterior, as opposed to quantifying the
``degrees of freedom'' in the prior. In other words,
infinite-dimensional Bayesian linear models with finite intrinsic
dimension are not identified with models for which the prior  is concentrated on a
finite-dimensional manifold of the infinite-dimensional state space. 

A similar consideration of balancing tractability and practical
relevance has dictated the choice not to study interacting particles
typically used for filtering, but rather to focus on one-step
filtering using importance sampling. For such problems we introduce
appropriate notions of intrinsic dimension and compare the relative
merits of popular alternative schemes.

The most pressing topic for future research stemming from this article
is the development of concrete recommendations for algorithmic design
within classes of Bayesian models used in practice. Within the model
structure we have studied here, practically relevant and important
extensions include models with non-Gaussian priors on the unknown,
nonlinear operators that link the unknown to the data, and unknown
hyperparameters involved in the model specification. Linearisation of
a nonlinear model around some reasonable value for the unknown
(e.g. the posterior mean) is one
way to extend our measures of intrinsic dimension in such
frameworks. We can expect the subject area to see considerable 
development in the coming decade.

\section*{Acknowledgments}
The authors are thankful to Alexandre Chorin, Arnaud Doucet, Adam Johansen, and Matthias Morzfeld for their generous feedback. SA and DSA are grateful to EPSRC for financial support. AMS is grateful to DARPA, EPSRC and ONR for financial support.

\bibliographystyle{plain}
\bibliography{isbib}

\section{Supplementary Material}

\subsection{Gaussian Measures in Hilbert Space}\label{ssec:Gaussian}

In section \ref{sec:BIP} we study Bayesian inverse problems in 
the Hilbert
space setting. This enables us to talk about infinite dimensional limits
of sequences of high dimensional inverse problems and is hence useful
when studying the complexity of importance sampling in high dimensions.
Here we provide some background on Gaussian measures in Hilbert space.
We start by describing how to construct a random draw from a Gaussian
measure on an infinite dimensional separable Hilbert space 
$(\h,\pr{\cdot}{\cdot},\norm{\cdot})$. 
Let $\excov:\h \to\h$ be a self-adjoint, positive-definite and trace class operator. It then holds that $\excov$ has a countable set of eigenvalues $\{\excoef_j\}_{j\in\N}$, with corresponding normalized eigenfunctions $\{e_j\}_{j\in\N}$ which form a complete orthonormal basis in $\h$.

\begin{example}
\label{ex:bb}
We use as a running example the case where $\h$ is the space of square 
integrable real-valued functions on the unit interval, $\h=L^2(0,1)$ and
where the Gaussian measure of interest is a unit centred
Brownian bridge on the interval $(0,1)$. Then $m=0$ and 
$\excov$ is the inverse of the negative
Laplacian on $(0,1)$ with homogeneous Dirichlet boundary conditions.
The eigenfunctions and eigenvalues of $\excov$ are given by
$$e_j(t)=\sqrt{2}\sin(j\pi t), \quad \kappa_j=(j\pi)^{-2}.$$
The eigenvalues are summable and hence the operator $\excov$
is trace class.  For further details see \cite{AS10}.
\end{example}

For any $\exm\in\h$, we can write a draw $x\sim\G(\exm,\excov)$ as  \begin{equation*}
x=\exm+\sumj\sqrt{\excoef_j}\zeta_je_j,\end{equation*}
where $\zeta_j$ are independent standard normal random variables in $\R$; this is the Karhunen-Loeve expansion \cite[Chapter III.3]{RA90}. The trace class assumption on the operator $\excov$, ensures that $x\in\h$ with probability $1$, see Lemma \ref{propositionappendix} below.  The particular rate of decay of the eigenvalues $\{\excoef_j\}$ determines the almost sure regularity properties of $x$. The idea is that the quicker the decay, the smoother $x$ is, in a sense which depends on the basis $\{e_j\}$.  
For example if $\{e_j\}$ is the Fourier basis, which is the case if $\excov$ is a function of the Laplacian on a torus, then a quicker decay of the eigenvalues of $\excov$ means a higher H\"older and Sobolev regularity (see \cite[Lemmas 6.25 \& 6.27]{AS10} and \cite[Section 2.4]{DS15}).
For the Brownian bridge Example \ref{ex:bb} above, 
draws are almost surely in
spaces of both H\"older and Sobolev regularity upto (but not including)
one half.

The above considerations suggest that we can work entirely in the 
``frequency'' domain, namely the space of coefficients of the
element of $\h$ in the eigenbasis of the covariance, the sequence
space $\ell^2$. Indeed, we can identify the Gaussian measure $\G(\exm,\excov)$ with the independent product measure $\bigotimes_{j=1}^\infty \G(\exm_j, \excoef_j),$ where $\exm_j=\pr{\exm}{e_j}$.  
Using this identification, we can define a sequence of Gaussian measures in $\R^d$ which converge to $\G(\exm,\excov)$ as $d\to\infty$, by truncating the product measure to the first $d$ terms. 
Even though in $\R^d$ any two Gaussian measures with
strictly positive covariances are absolutely continuous with respect to each other (that is, equivalent as measures), in the infinite-dimensional limit two Gaussian measures can be mutually singular, and indeed are unless very stringent conditions are satisfied.

For  $\G(\exm, \excov)$ in $\h$, we define its Cameron-Martin space $E$ as the domain of $\excov^{-\frac12}.$ This space can be characterized as the space of all the shifts in the mean which result in an equivalent Gaussian measure, whilst the covariance is fixed. Since $\excov$ is a trace class operator, its inverse (hence also its square root) is an unbounded operator, therefore $E$ is a compact subset of $\h$. In
fact $E$ has zero measure under $\G(0, \excov)$. For example, if $\excov$ is given by the Brownian bridge Example \ref{ex:bb}, 
then the Cameron-Martin space $E$ is the Sobolev space of functions which vanish on the boundary and whose first derivative is in $\h$; in contrast, and as mentioned above, draws from this measure only have upto half a derivative in the Sobolev sense. The equivalence or singularity of two Gaussian measures with different covariance operators and different means depends on the compatibility of both their means and covariances, as expressed in the three conditions of the Feldman-Hajek theorem. For more details on the equivalence and singularity of Gaussian measures see \cite{DPZ92}. %In particular, two Gaussian measures $\G(\exm_1,\excov), \G(\exm_2, \excov)$, are equivalent if and only if their means are compatible in the sense that $\exm_1-\exm_2\in \D(\excov^{-\frac12})$.

The Karhunen-Loeve expansion makes sense even if $\excov$ is not trace class, in which case it defines a Gaussian measure in a space $\X\supset\h$ with a modified covariance operator which \emph{is} trace class. Indeed, let $D:\h\to\h$ be any injective bounded self-adjoint operator such that: a) $D$ is diagonalizable in $\{e_j\}_{j\in\N}$, with (positive) eigenvalues $\{d_j\}_{j\in\N}$; b) the operator $D\excov D$ is trace class, that is, $\{\excoef_jd_j^2\}_{j\in\N}$ is summable. Define the weighted inner product $\pr{\,\cdot\,}{\,\cdot\,}_{D^{-2}}:=\pr{D\,\cdot\,}{D\,\cdot}$, the weighted norm $\norm{\,\cdot\,}_{D^{-2}}=\norm{D\,\cdot\,}$ and the space 
$$\X\coloneq\overline{{\rm span} \{e_j:j\in\N\}}^{\norm{\,\cdot\, }_{D^{-2}}}.$$ 
Then the functions $\psi_j=d_j^{-1}e_j, \;j\in\N$, form a complete orthonormal basis in the Hilbert space $(\X,\pr{\,\cdot\,}{\,\cdot\,}_{D^{-2}},\norm{\,\cdot\,}_{D^{-2}})$. The Karhunen-Loeve expansion can then be written as 
\begin{equation*}
x=\exm+\sumj\sqrt{\excoef_j}\zeta_je_j=\exm+\sumj\sqrt{\excoef_j}d_j\zeta_j\psi_j,\end{equation*}
so that we can view $x$ as drawn from the Gaussian measure $\G(\exm,D\excov D)$ in $\X$, where $D\excov D$ is trace class by assumption. For example, the case $\h=L^2(0,1)$ and $\excov=I$, corresponding to Gaussian white 
noise for functions on the interval $(0,1)$, 
can be made sense of in negative Sobolev-Hilbert spaces with $-1/2-\epsilon$ derivatives, for any $\epsilon>0$. Finally, we stress that absolute continuity in general and the Cameron-Martin space in particular, are concepts which are independent of the space in which we make sense of the measure. In the Gaussian white noise example, we hence have that the Cameron-Martin space is $E=\h$.

The following lemma is similar to numerous results concerning
Gaussian measures in function spaces. Because the precise form 
which we use is not in the literature, we provide a direct proof.

\begin{lemma}\label{propositionappendix}
Let $\X$ be a separable Hilbert space with orthonormal basis $\{\varphi_j\}_{j\in\N}$. Define the Gaussian measure $\gamma$ through the Karhunen-Loeve expansion \[\gamma:=\mathcal{L}\Bigl(\sumi \sqrt{\lambda_j}\xi_j\varphi_j\Bigr),\]
where $\lambda_j$ is a sequence of positive numbers and where $\xi_j$ are i.i.d. standard normal.
Then draws from $\gamma$ are in $\X$ almost surely if and only if $\sumi\lambda_j<\infty$.
\end{lemma}

\begin{proof}
If $\sumi\lambda_j<\infty$, then
\[\E_\gamma\norm{x}_{\X}^2=\E\sumi\lambda_j\xi_j^2=\sumi\lambda_j<\infty,\]
hence $x\sim\gamma$ is in $\X$ almost surely. 

For the converse, suppose that $x\sim\gamma$ is in $\X$ almost surely. Then 
 \[\norm{x}_{\X}^2=\sumi\lambda_j\xi_j^2<\infty, \quad \text{a.s.}\]
 Note that this implies that $\lambda_j\to0$, and so in particular $\lambda_\infty:=\sup_{j} \lambda_j<\infty.$

 By \cite[Theorem 3.17]{OK02}, since $\sqrt{\lambda_j}\xi_j\sim N(0,\lambda_j)$ are independent and symmetric random variables,
we get that \[\sumi\E[\lambda_j\xi_j^2\wedge1]<\infty.\]

A change of variable gives
\begin{align*}
\E[\lambda_j\xi_j^2\wedge1]&\ge\frac{2}{\sqrt{2\prm\lambda_j}}\int_0^1y^2e^{-\frac{y^2}{2\lambda_j}}dy \\
&=\frac{2\lambda_j^{\frac32}}{\sqrt{2\prm\lambda_j}}\int_0^{1/\sqrt{\lambda_j}}z^2e^{-\frac{z^2}2}dz=\frac{2\lambda_j}{\sqrt{2\prm}}\int_0^{1/\sqrt{\lambda_j}}z^2e^{-\frac{z^2}2}dz.
\end{align*}
Thus, for every $j\in \N,$ 
 \[\E[\lambda_j\xi_j^2\wedge1]\geq \frac{2\lambda_j}{\sqrt{2\prm}} \int_0^{1/\sqrt{\lambda_\infty}}z^2e^{-\frac{z^2}2}dz.\] Since the left hand side is summable, we conclude that \[\sumi\lambda_j<\infty.\]
\end{proof}

\subsection{Details of the inverse problem setting in section \ref{sec:BIP}}
In section \ref{sec:BIP} we assume Gaussian observation noise  $\eta\sim\noise:=\G(0, \, \nc)$ and put a Gaussian prior  on the unknown $u \sim\prior=\G(0,\, \prc)$, where  $\nc:\h\to\h$ and $\prc:\h\to\h$ are bounded, self-adjoint, positive-definite linear operators. 
As discussed in subsection  \ref{ssec:Gaussian},
if the covariance $\nc$ (respectively $\prc$) is trace class then 
$\eta \sim \noise$ (respectively $u \sim \prior$) is almost surely
in $\h$. On the other hand, as also discussed in 
subsection \ref{ssec:Gaussian}, when {the} covariance $\nc$
(respectively $\prc$) is not 
trace-class we have that $\eta \notin \h$ but  $\eta\in \Y$ $\noise$-almost
surely (respectively $u\notin \h$ but $u \in \X$ $\prior$-almost surely) 
where $\Y$ (respectively $\X$) strictly contains $\h$; 
indeed $\h$ is compactly embedded into $\X,\Y.$

We tacitly assume that $\fp$ can be extended to act on elements in $\X$ and that the sum of $\fp u$ and $\eta$ makes sense in $\Y$. This assumption holds trivially if the three operators $\fp, \prc, \nc$ are simultaneously diagonalizable as in Example \ref{ex:lm11}. It also holds in non-diagonal settings, in which it is possible to link the domains of powers of the three operators by appropriate embeddings; for some examples see \cite[Section 7]{ALS13}.

%\marginpar{comment that since the cov is assumed to be bounded, can make sense of non-trace class operator without any reference to the particular covariance} %I can comment more on this and provide a proof: $\excov^{\frac12}(\h)=(Q\excov Q)^\frac12(\X)$

\subsection{Proofs Section \ref{sec:IS}}
\label{ssec:PIS} \label{P21}
Throughout we denote by  $\mc$  the empirical random measure 
\begin{equation*}
\mc:=\frac{1}{N}\sum_{n=1}^N \delta_{\un}, \quad \un\sim \propo.
\end{equation*}
We recall that $\mu^N$ denotes the particle approximation of $\tm$ based on sampling from the proposal $\pi.$

\subsubsection{Proof of Theorem \ref{non-asymptotictheorem}}
\label{sssec:t2.1}

\begin{proof}[Proof of Theorem \ref{non-asymptotictheorem}]
For the bias we write 
\begin{align*}
\is(\test) - \tm(\test)&=\frac{1}{\Zn}\mc(\test\rn)   -   \tm(\test)\\
&=\frac{1}{\Zn} \mc\Bigl(\bigl(\test-\tm(\test)\bigr) \rn\Bigr).
\end{align*}
Then, letting $\testc:=\test-\tm(\test)$
and noting that
$$\propo(\testc \rn)=0$$ 
we can rewrite
\begin{equation*}
\is(\test) -  \tm(\test)=\frac{1}{\Zn}\Bigl( \mc(\testc \rn)- \propo( \testc \rn)\Bigr).
\end{equation*}
The first of the terms in brackets is an unbiased estimator of the second one, and so
\begin{align*}
\E \bigl[\is(\test) - \tm(\test)\bigr]&=\E\Biggl[\Bigl(\frac{1}{\Zn}-\frac{1}{\const}\Bigr)\Bigl( \mc(\testc \rn)- \propo( \testc \rn)\Bigr)\Biggr]\\
&=\E\Biggl[\frac{1}{\Zn\const}\Bigl(\const-\Zn\Bigr)\Bigl( \mc(\testc \rn)- \propo( \testc \rn)\Bigr)\Biggr].
\end{align*}

Therefore,   \begin{align*}
\Bigl|\E &\bigl[\is(\test) - \tm(\test)\bigr]\Bigr|\\
& \le \biggl|\E \Bigl[\bigl(\is(\test) - \tm(\test)\bigr)1_{\{2\Zn>\const\}}\Bigr]\biggr| + \biggl|\E \Bigl[\bigl(\is(\test) - \tm(\test)\bigr)1_{\{2\Zn \le \const\}}\Bigr]\biggr|\\
&\le \frac{2}{\const^2}\E\Bigl[\bigl|\const-\Zn\bigr|  \bigl| \mc(\testc \rn)- \propo( \testc \rn)\bigr|\Bigr] + 2\rp\Bigl(2\Zn \le \const\Bigr) \\
&\le \frac{2}{\const^2} \frac{1}{\sqrt{N}}\propo\bigl(\rn^2\bigr)^{1/2} \frac{2}{\sqrt{N}}\propo\bigl(\rn^2\bigr)^{1/2} + 2\rp\Bigl(2\Zn\le \const\Bigr),
\end{align*}where in the second and third inequality we used that $|\test|\leq1$.
Now note that
$$\rp\Bigl(2\Zn\le \const\Bigr)=\rp\Bigl(2(\Zn-\const)\le -\const\Bigr)
\le \rp\Bigl(2|\Zn-\const|\ge \const\Bigr).$$
By the Markov inequality $\rp\Bigl(2\Zn\le \const\Bigr)\le \frac{4}{N}\frac{\propo(\rn^2)}{\const^2},$ and so
\begin{equation*}
\sup_{|\test|\le 1}\Bigl|\E \bigl[\is(\test) - \tm(\test)\bigr]\Bigr| \le \frac{12}{N}\frac{\propo(\rn^2)}{\const^2}.
\end{equation*}
This completes the proof of the result for the bias. For the MSE 
\begin{align}\label{eq:trick}
\is(\test)-\tm(\test) &= \frac{1}{\Zn}\mc(\test\rn) - \frac{1}{\const}\propo(\test\rn)\nonumber \\
&= \biggl( \frac{1}{\Zn} - \frac{1}{\const} \biggr) \mc(\test\rn) -  \frac{1}{\const}\Bigl( \propo(\test\rn)-\mc(\test\rn)     \Bigr) \nonumber\\
&= \frac{1}{\const}\Bigl(\const - \Zn\Bigr) \is(\test) -  \frac{1}{\const}\Bigl( \propo(\test\rn)-\mc(\test\rn)     \Bigr),
\end{align}
and so using the inequality $(a+b)^2 \leq 2(a^2+b^2)$ we obtain
\begin{equation*}
\bigl(\is(\test)-\tm(\test) \bigr)^2 \le \frac{2}{\const^2} \biggl\{ \Bigl(\const - \Zn\Bigr)^2 \is(\test)^2 + \Bigl( \propo(\test\rn)-\mc(\test\rn)     \Bigr)^2   \biggr\}. 
\end{equation*}
Therefore, for $|\test|\le 1,$
\begin{align*}
\E \left[  \bigl(\is(\test)-\tm(\test) \bigr)^2  \right]& \le  \frac{2}{\const^2}  \Biggl\{ \E\biggl[ \Bigl(\const - \Zn\Bigr)^2 \biggr] + \E\biggl[\Bigl( \propo(\test\rn)-\mc(\test\rn)     \Bigr)^2\biggr]   \Biggr\}  \\
&= \frac{2}{\const^2}  \Bigl\{  \Var_{\propo}\bigl(\Zn\bigr)  + \Var_{\propo}\bigl(\mc(\test\rn)\bigr)  \Bigr\}  \\
&\le  \frac{2}{N\const^2}  \Bigl\{  \propo\bigl(\rn^2\bigr) +  \propo\bigl(\test^2 \rn^2\bigr) \Bigr\} \\
&\le \frac{4}{N}\frac{\propo\bigl(\rn^2\bigr)}{\const^2},
\end{align*}
and the proof is complete.
\end{proof}
\begin{remark}
The constant $12$ for the bias can be somewhat  reduced by using in the proof the indicator $1_{\{a\Zn \le \const\}}$ instead of $1_{\{2\Zn \le \const\}}$ and optimizing over $a>0$. Doing this yields the constant $C\approx 10.42 $ rather than $C=12.$
\end{remark}

\subsubsection{Proof of Theorem \ref{thm:nonasmom}} \label{P23}
The proof of the MSE part of Theorem \ref{thm:nonasmom} uses the approach of \cite{DL09} for calculating moments of ratios of estimators. The proof of the bias part is very similar to the proof of the bias part of Theorem \ref{non-asymptotictheorem}.

In order to estimate the MSE, we use \cite[Lemma 2]{DL09} which in our setting becomes:
\begin{lemma}\label{dlem1}
For $0<\theta<1$, it holds 
\begin{align*}
\bigl|\is(\test)-\tm(\test)\bigr|\leq \frac{|\mc(\test\rn)-\propo(\test\rn)|}{\const}+&\frac{|\mc(\test\rn)|}{\Zpi^2}|{\Zn}-\Zpi| \\
&\quad\quad+\max_{1\leq n\leq N}|\test(\un)|\frac{|{\Zn}-\Zpi|^{1+\theta}}{\Zpi^{1+\theta}}.
\end{align*}

\end{lemma}

The main novelty of the above lemma compared to the bounds we used in the proof of Theorem \ref{non-asymptotictheorem}, is not the bound on $\test$ using the maximum, but rather the introduction of $\theta\in(0,1)$. This will be apparent in the proof of Theorem \ref{thm:nonasmom} below.

We also repeatedly use  H\"older's inequality in the form \[\E\bigl[|uv|^s\bigr]\leq\E\bigl[|u|^{sa}\bigr]^{\frac1{a}}\E\bigl[|v|^{sb}\bigr]^{\frac1{b}},\] for any $s>0$ and for $a,b>1$ such that $\frac1a+\frac1b=1$, as well as the Marcinkiewicz-Zygmund inequality \cite{RL01}, which for centered  i.i.d. random variables $X_n$ gives  
\[\E\left[\Bigl|\sumn X_n\Bigr|^t \right]\leq C_tN^{\frac{t}2}\E\bigl[|X_1|^t \bigr], \quad \forall t\geq2.\]
There are known bounds  on the constants, namely  $C_t^\frac1t\leq t-1$, \cite{RL01}. We apply this inequality in several occasions with $X_n=h(\un)-\propo(h)$ for different functions $h$, in which case we get 
\begin{align}\label{eq:MZc}
\E\Bigl[\bigl|\mc(h)-\propo(h)\bigr|^t \Bigr]\leq C_t\E\Bigl[\bigl|h(u^1)-\propo(h)\bigr|^t \Bigr]N^{-\frac{t}2}, \quad \forall t\geq 2.
\end{align}
%or
%\begin{align}\label{eq:MZ}
%\norm{\mc(h)-\propo(h)}_t\leq \tilde{C}_t\norm{h}_tN^{-\frac12},
%\end{align}
%where $\tilde{C}_t$ appears because we bound the central by the noncentral moment of $h$ (is this OK?)
%%%%
We are now ready to prove Theorem \ref{thm:nonasmom}.

%Note that the novelty of the bound in the lemma, is the introduction of $\alpha\in(0,1)$. If instead of introducing $\alpha$, we attempt to directly use \eqref{eq:1}, we run into trouble because whatever rate of convergence we get from ${\Zn}-Z$, is partially cancelled out by the bound that we use for the maximum (see proof of theorem below, we bound the max by the sum) and so we get suboptimal rates.

\begin{proof}[Proof of Theorem \ref{thm:nonasmom}]
We first prove the MSE part. 
By Lemma \ref{dlem1} we have that 
\[\E\Bigl[\bigl(\is(\test)-\tm(\test)\bigr)^2\Bigr]\leq 3A_1+3A_2+3A_3,\]
where $A_1, A_2, A_3$ correspond to the second moments of the three terms respectively.
\begin{enumerate}
\item For the first term we have \[A_1=\frac1{\const^2}\E\biggl[\Bigl(\mc(\test\rn)-\propo(\test\rn)\Bigr)^2\biggr]\leq \frac{1}{\const^2}\E\biggl[\Bigl(\test(u^1)\rn(u^1)-\propo(\test\rn)\Bigr)^2\biggr]N^{-1}.\]

\item For the second term, H\"older's inequality gives 
\begin{align*}
A_2&=\frac{1}{\const^4}\E\Bigl[\bigl|\mc(\test\rn)\bigl(\Zn-\const\bigr)\bigr|^2\Bigr]\\
&\leq \frac1{\const^4}\E\Bigl[\bigl|\mc(\test\rn)\bigr|^{2d}\Bigr]^\frac1d\E\Bigl[\bigl|\Zn-\const\bigr|^{2e}\Bigr]^\frac1e,
\end{align*}
where $\frac1d+\frac1e=1$. Use of the triangle inequality yields
\begin{align*}
\E\Bigl[\bigl|\mc(\test\rn)\bigr|^{2d}\Bigr]^\frac1d&=\frac{1}{N^2}\E\biggl[\Bigl|\sumn \test(\un)\rn(\un)\Bigr|^{2d}\biggr]^\frac1d \\
&\leq \propo\bigl(|\test\rn|^{2d}\bigr)^\frac1d.
\end{align*}
Combining with \eqref{eq:MZc} (note that $t=2e>2$) we get 
\[A_2\leq \frac{1}{\Zpi^4}\propo\bigl(|\test\rn|^{2d}\bigr)^\frac1d{C}_{2e}^\frac1e\E\Bigl[\bigl|\rn(u_1)-\propo(\rn)\bigr|^{2e}\Bigr]^\frac1eN^{-1}.\]

\item By H\"older we have 
\begin{align*}A_3&=\frac{1}{\const^{2(1+\theta)}}\E\Biggl[\max_{1\leq n\leq N}|\test(\un)|^2\bigl|\const-\Zn\bigr|^{2(1+\theta)}\Biggr]\\
&\leq \frac{1}{\const^{2(1+\theta)}}\E\Biggl[\max_{1\leq n\leq N}\bigl|\test(\un)\bigr|^{2p}\Biggr]^\frac1p\E\Biggl[\bigl|\const-\Zn\bigr|^{2q(1+\theta)}\Biggr]^\frac1{q},
\end{align*}
where $\frac1p+\frac1q=1$. 
Note that 
\[\E\Biggl[\max_{1\leq n\leq N}\bigl|\test(\un)\bigr|^{2p}\Biggr]^\frac1p\leq \E\Biggl[\sumn\bigl|\test(\un)\bigr|^{2p}\Biggr]^\frac1p=N^\frac1p\propo\bigl(|\test|^{2p}\bigr)^\frac1p.\]
Combining with \eqref{eq:MZc}, with  $t_\theta=2q(1+\theta)>2$, we get 
\[A_3\leq \frac{1}{\const^{2(1+\theta)}}N^\frac1p\propo\bigl(|\test|^{2p}\bigr)^\frac1p{C}_{t_\theta}^\frac1q\E\Bigl[\bigl|\rn(u^1)-\propo(\rn)\bigr|^{t_\theta}\Bigr]^\frac1qN^{-1-\theta}.\]
Now choosing $\theta=\frac1p\in(0,1)$ gives the desired order of convergence
\[A_3\leq \frac{1}{\const^{2(1+\frac1p)}}\propo(|\test|^{2p})^\frac1p{C}_{2q(1+\frac1p)}^\frac1q\E\Bigl[\bigl|\rn-\propo(\rn)\bigr|^{2q(1+\frac1p)}\Bigr]^\frac1qN^{-1}.\] 
\end{enumerate}

This completes the proof of the MSE part. For the bias, as in the proof of Theorem \ref{non-asymptotictheorem} we have
 \begin{align*}
 \Bigl|\E &\bigl[\is(\test) - \tm(\test)\bigr]\Bigr|\\
&\le \frac{2}{\const^2}\E\biggl[\Bigl|\const-\Zn\Bigr|  \Bigl| \mc(\testc \rn)- \propo( \testc \rn)\Bigr|\biggr] +  \biggl|\E \Bigl[\bigl(\is(\test) - \tm(\test)\bigr)1_{\{2\Zn \le \const\}}\Bigr]\biggr|,
\end{align*}
where $\testc=\test-\tm(\test)$. Using the Cauchy-Schwarz inequality
we obtain
 \begin{align*}
 \Bigl|\E &\bigl[\is(\test) - \tm(\test)\bigr]\Bigr|&\\
&\leq \frac{2}{\const^2}\E\Bigl[\bigl|\const-\Zn\bigr|^2\Bigr]^\frac12\E\Bigl[\bigl| \mc(\testc \rn)- \propo( \testc \rn)\bigr|^2\Bigr]^\frac12&\\ &\quad\quad\quad\quad\quad\quad+ \E\Bigl[\bigl(\is(\test)-\tm(\test)\bigr)^2\Bigr]^\frac12 \rp\Bigl(2\Zn\le \const\Bigr)^\frac12&\\
&\leq\frac{2}{\const^2}\frac{1}{N}\E\Bigl[\bigl|\rn(u^1)-\propo(\rn)\bigr|^2\Bigr]^\frac12\E\Bigl[\bigl|\testc(u^1)\rn(u^1)-\propo(\testc\rn)\bigr|^2\Bigr]^\frac12 + \frac{C_{\rm MSE}^\frac12}{N^\frac12}\frac{2}{N^\frac12}\frac{\propo(\rn^2)^\frac12}{\const},&
\end{align*}
where to bound the probability of $2\Zn\leq \const$ we use 
the Markov inequality similarly as in 
the analogous part of the 
proof of Theorem \ref{non-asymptotictheorem}.
\end{proof}

%%%%
\subsection{Proofs Section \ref{sec:BIP}} \label{ap:bip}
We next state a lemma  collecting several useful properties of the trace of linear operators. A compact linear operator $T$ is said to belong in the trace class family if its singular values $\{\sigma_i\}_{i=1}^\infty$ are summable. In this case we write $\tr(T)=\sum_{i=1}^\infty \sigma_i$, while for notational convenience we define the trace even for operators that are not trace class, with infinite value.  $T$ is said to belong in the Hilbert-Schmidt family, if its singular values are square summable (equivalently if $T^\ast T$ is trace class). 

\begin{lemma}\label{lem:tr}
Let $T$ be an operator on a Hilbert space $\h$. Suppose for the next three items that $T$ is trace class. Then
\begin{enumerate}
\item[i)]$\tr(T^\ast)=\overline{\tr(T)}$. In particular, if the eigenvalues of $T$ are real then $\tr(T^\ast)=\tr(T)$;
\item[ii)]for any bounded operator $B$ in $\h$, $\tr(TB)=\tr(BT)$ and this assertion also holds if $T$ and $B$ are Hilbert-Schmidt;
\item[iii)]for any bounded operator $B$ in $\h$, $\tr(TB)=\tr(BT)\leq\norm{B}\tr(T)$.
\end{enumerate}

For any bounded linear operator $T$, it holds that
\begin{enumerate}
\item[iv)]  $\tr(T^\ast T) = \tr(TT^\ast),$
\end{enumerate}
where if $T$ (equivalently $T^\ast$) is not Hilbert-Schmidt, we define the trace to be $+\infty$. 

If $T$ is a linear operator and $P$ is bounded and positive definite, such that $TP^{-1}$ (equivalently $P^{-\frac12}TP^{-\frac12}$ or $P^{-1}T$) is bounded, it holds that 
\begin{enumerate}
\item[v)]$\tr(TP)=\tr(P^{\frac12}TP^{\frac12})=\tr(PT),$
\end{enumerate}
where as in (iv) we allow infinite values of the trace.

Finally, suppose that $D_1$ is positive definite and $D_2$ is positive semi definite, and that $T$ is self adjoint and bounded in $\h$. Furthermore, assume that $D_1^{-1}T$ and $(D_1+D_2)^{-1}T$ have  eigenvalues. Then 
\begin{enumerate}
\item[vi)]  $\tr(D_1^{-1} T) \ge \tr\bigl( (D_1+D_2)^{-1} T).$
\end{enumerate}
\begin{comment}
Let $Q, R$ be bounded operators in $\h$, with $R$ invertible. Then 
\begin{enumerate}
\item[vi)]  $\tr(Q^\ast R^{-1}Q)=\tr(R^{-1}QQ^\ast).$
\end{enumerate}
\end{comment}
\end{lemma}
\begin{proof}
The proofs of parts (i)-(iii) can be found in \cite[Section 30.2]{PDL02}, while (iv) is an exercise in \cite[Section 30.8]{PDL02}. Part (v) can be shown using the infinite-dimensional analogue of matrix similarity, see \cite[Section 2]{AHV82}. In particular, if we multiply $TP$ to the left by $P^{1/2}$ and to the right by $P^{-1/2}$, we do not change its eigenvalues hence neither its trace, so $\tr(TP)=\tr(P^{\frac12}TP^{\frac12})$. Similarly, if we multiply $TP$ to the left by $P$ and to the right by $P^{-1}$, we get $\tr(TP)=\tr(PT)$.
Part (vi) follows from the stronger fact that the ordered eigenvalues of $D_1^{-1} T$ are one by one bounded by the ordered eigenvalues of $(D_1+D_2)^{-1} T$. This in turn can be established using that the eigenvalues of these operators are determined by the generalized eigenvalue problem $Tv = \lambda D_1v$ and $Tv = \lambda (D_1+D_2) v,$ with associated Rayleigh quotients
\begin{equation}
\frac{\langle x, Tx \rangle}{\langle x, D_1 x\rangle} \ge \frac{\langle x, Tx \rangle}{\langle x, (D_1+D_2) x\rangle},
\end{equation}
and an application of the Rayleigh-Courant-Fisher theorem (see \cite{PDL02} and \cite{RS78}).
%\begin{comment} Finally, for part (vi) we first use part (iv), which we notice that it holds even if $T$ is not Hilbert-Schmidt (since then both sides are infinite) to get that \[\tr(Q^\ast R^{-1}Q)=\tr(R^{-\frac12}QQ^\ast R^{-\frac12}).\] In order to get the desired result we notice that for $T=R^{-\frac12}QQ^\ast R^{-\frac12}$, it suffices to have that $T$ and $R^{-\frac12}TR^{\frac12}$ have the same spectra. Indeed, this holds by the infinite-dimensional analogue of matrix similarity, see \cite[Section 2]{AHV82}. \end{comment}
\end{proof} 
\label{sec:PBIP}

\subsubsection{Proofs of subsection \ref{ssec:ess2}} \label{PSS32}

\begin{proof} [Proof of Proposition \ref{prop:new}]

We give the proof for the final dimensional case. For the extension to infinite dimensions see Remark \ref{rem:propinfdim} below.
Under the given assumptions, expression \eqref{eq:fdposc} for $C^{-1}$
is well-defined and gives
\begin{equation}
\label{eq:one}
\prc^{\frac12} \posc^{-1}\prc^{\frac12}=I+A.
\end{equation}
Thus
\begin{align*}
\tr(A)&=\tr(\posc^{\frac12} \posc^{-1}\prc^{\frac12}-I)\\ 
&=\tr\bigl(\posc^{\frac12}(\posc^{-1}-\prc^{-1})\prc^{\frac12}\bigr)\\
&=\tr\bigl((\posc^{-1}-\prc^{-1})\prc\bigr),
\end{align*}
where the last equality is justified using the cyclic property of the trace, Lemma \ref{lem:tr}(ii). For the second identity, since $(I+A)^{-1}A=I-(I+A)^{-1}$, we have again by (\ref{eq:one})  \begin{align*}
\tr((I+A)^{-1}A)&=\tr\Bigl(I-(I+A)^{-1}\Bigr)\\
&= \tr\Bigl(I - \prc^{-1/2}
\posc \prc^{-1/2}\Bigr)\\
&= \tr \Bigl( \prc^{-1/2}
( \prc - \posc)  \prc^{-1/2}\Bigr) 
\\
&=\tr\Bigl( (\prc - \posc) \prc^{-1}\Bigr),
\end{align*}
where the last equality is again justified via the cyclic property of the trace. 
\end{proof}

\begin{remark}\label{rem:propinfdim}
Proposition \ref{prop:new} also holds in the Hilbert space setting, but requires formula \eqref{eq:fdposc} for the precision operator of the posterior is justified see Remark \ref{remarkfininf} and \cite[Section 5]{ALS13}.
Indeed, the proofs of the two identities are almost identical to the finite dimensional case, the only difference being in the justification of the last equalities in the two sequences of equalities above. In this case the two trace-commutativity equalities have to be justified using Lemma \ref{lem:tr}(v) rather than Lemma \ref{lem:tr}(ii).  In the first case, Lemma \ref{lem:tr}(v) can be applied, since $A=\prc^{\frac12}(\posc^{-1}-\Sigma^{-1})\prc^\frac12$ is bounded by Assumption \ref{Generalinvprobassumption}, and $\prc$ is assumed to be positive definite and bounded.
In the second case, Lemma \ref{lem:tr}(v) can be applied, since by Assumption \ref{Generalinvprobassumption} the operator $(I+A)^{-1}A$ is bounded, and $\prc$ is bounded and positive definite.
\end{remark}

\begin{proof} [Proof of Proposition \ref{prop:efd}]

\begin{enumerate}
\item We have that $(v_i,\mu_i)$ is an eigenvector/value pair of the first matrix
  if and only if $(\Gamma^{-1/2} v_i,\mu_i)$ is of the second. It is also
  immediate that $(v_i,\mu_i)$ is a pair for the second if and
  only if 
  $(\auxm^\ast v_i,\mu_i)$ is for $A(I + A)^{-1}$. However, it is
  also easy to check that $A(I + A)^{-1} = (I +
  A)^{-1}A$. 
\item In view of the above, note that $(v_i,\mu_i)$ is a pair for $(
  I+A)^{-1}A$ if an only if $(v_i, \mu_i / (1-\mu_i))$
  is for $A$. Hence, if $\lambda_i$ is an eigenvalue of $A$,
  $\lambda_i/(1+\lambda_i)$ is one for the other matrices. Given
  that this is always less or equal to 1 and the $\effd$ is a trace of
  either $d_y \times d_y$ or $d_u \times d_u$ matrices, the inequality
  follows immediately.
\end{enumerate}
\end{proof}

\begin{proof} [Proof of Lemma \ref{propequivefd}]
 If $A$ is trace class then it is compact and since it is also self-adjoint and nonnegative it can be shown (for example using the spectral representation of $A$) that $\norm{(I+A)^{-1}}\leq 1.$ Then Lemma \ref{lem:tr}(iii) implies that 
\[\tr\bigl((I+A)^{-1}A\bigr)\leq \tr(A).\]

Assume now that $(I +
 A )^{-1} A$  is trace class. Then $A$ is too since it is the product
 of the bounded operator $ I + A$ and the trace class operator $(I +
 A )^{-1} A$, see again Lemma \ref{lem:tr}(iii). In particular, 
\[
\tr(A)  \leq \|I + A\| \tr\bigl(( I +
 A )^{-1} A\bigr).
\]
\end{proof}

\subsubsection{Proofs of subsection \ref{ssec:con2}}  \label{PSS33}

 \begin{proof} [Proof of Theorem \ref{maintheorem}]

$i) \Leftrightarrow ii)$ is immediate from Lemma \ref{propequivefd}. \\
$ii) \Leftrightarrow iii)$ It holds that $\nc^{-\frac12}\fp u\sim N(0, \nc^{-\frac12}\fp\prc \fp^\ast\nc^{-\frac12})$ since $\nc^{-\frac12}\fp u$ is a linear transformation of the Gaussian $u\sim\prior=N(0,\prc)$. 
By Lemma \ref{propositionappendix}  and since $A$ has eigenvalues, we hence have that $\nc^{-\frac12}\fp u\in \h$ if and only if $\tr(\nc^{-\frac12}\fp\prc\fp^\ast\nc^{-\frac12})<\infty$. \\

$iii) \Rightarrow iv)$ 
According to the discussion in subsection \ref{ssec:Gaussian} on the absolute continuity of two Gaussian measures with the same covariance but different means, the Gaussian likelihood measure $\like=\G(\fp u,\nc)$ and the Gaussian noise measure $\noise=\G(0,\nc)$ are equivalent if and only if $\nc^{-\frac12}\fp u\in\h$. Under $iii)$, we hence have that $\like$ and $\noise$ are equivalent for $\propo$-almost all $u$ and under the Cameron-Martin formula \cite{DPZ92} for $\propo$-almost all $u$ we have
\begin{equation*}
\frac{d\like}{d\noise}(y)= \exp\left(-\frac{1}{2}\norm{{\nc^{-1/2}}\fp u}^2+ \pr{{\nc^{-1/2}}y}{{\nc^{-1/2}}\fp u}\right)=:\rnip(u;y).
\end{equation*}
Defining the measure $\nu_0(u,y):=\propo(u)\times\noise(y)$ in $\X\times \Y$, we then immediately have that  
\begin{equation*}
\frac{d\nu}{d{\nu_0}}(u,y)= \rnip(u;y),
\end{equation*}
where $\nu$ is the joint distribution of $(u,y)$ under the model
$y=Ku+\eta$ with $u$ and $\eta$ independent Gaussians $\G(0,\Sigma)$
and $\G(0,\Gamma)$ respectively.

We next show that $\propo(g(\cdot;y))>0$ for $\rp_\eta$-almost all $y$, which will in turn enable us to use a standard conditioning result to get that the posterior is well defined and absolutely continuous with respect to the prior. Indeed, it suffices to show that $g(u;y)>0$ $\nu_0$-almost surely. Fix $u\sim\propo$. Then, as a function of $y\sim\noise$ the negative exponent of $g$ is distributed as $\G(\frac12\|\nc^{-\frac12}\fp u\|^2, \|\nc^{-\frac12}\fp u\|^2)$ where $\|\nc^{-\frac12}\fp u\|^2<\infty$ with $\propo$ probability $1$. Therefore, for $\nu_0$-almost all $(u,y)$ the exponent is finite and thus $g$ is $\nu_0$-almost surely positive implying that $\propo(g(\cdot;y))>0$ for $\noise$-almost all $y$. Noticing that the equivalence of $\nu$ and $\nu_0$ implies the equivalence of the marginal distribution of the data under the model, $\mar$, with the noise distribution $\noise$, we get that $\propo(g(\cdot;y))>0$ for $\mar$-almost all $y$.
%Note that under $\nu_0$, it holds that $-\log(\rnip(u;y))\sim\G($ the assumption on the data \textbf{OP:  be explicit about assump on data}, $\langle
%\nc^{-\frac12} y, \nc^{-\frac12}\fp u \rangle$ is a real Gaussian random variable $\G\bigl(\langle
%\nc^{-\frac12} \fp u^\dagger, \nc^{-\frac12}\fp u\rangle, \|\nc^{-\frac12}\fp u\|^2\bigr)$ which is finite
%$\prior$-a.s. Therefore, for $\Q$-a.e. $y$ of the form \eqref{data
%  form} $0<\rnip(u,y)<\infty$ $\prior$-a.s., from where it follows that
%$0<\prior\bigl(\rnip(\cdot,y)\bigr)<\infty.$ \textbf{OP: this is for
%  a.e. $y$: the idea is that since integrated over $u,y$ it should
%  give 1, it cannot be that when integrated over $u$ yields infty for
% $y$ in non-trivial sets, right?}
Hence, we can apply Lemma 5.3 of \cite{hairer2007analysis}, to get that the posterior measure $\post(\cdot)=\nu(\cdot|y)$ exists $\mar$-almost surely and is given by 
\begin{equation*}
\frac{d\tm}{d\propo}(u)=\frac{1}{\const}\exp\left(-\frac{1}{2\ns}\norm{{\nc^{-1/2}}\fp u}^2+\frac{1}{\ns}\pr{{\nc^{-1/2}}y}{{\nc^{-1/2}}\fp u}\right).
\end{equation*} 

Finally, we note that since $\frac{d\nu}{d\nu_0}=\rnip$, we have that $\int_{\X\times\Y} g\,d\nu_0(u,y)=1.$ Thus the Fubini-Tonelli theorem implies that $\propo(g(\cdot;y))<\infty$ for $\noise$-almost all $y$ 
and hence also for $\mar$-almost all $y$.

$iv) \Rightarrow ii)$ Under $iv)$ we have that the posterior measure $\tm$ which, as discussed in subsection \ref{ssec:gen2}, is Gaussian with mean and covariance given by \eqref{eq:posm} and \eqref{eq:posc}, is $y$-almost surely absolutely continuous with respect to the prior $\propo=N(0,\prc)$. By the Feldman-Hajek theorem \cite{DPZ92}, we hence have that $y$-almost surely the posterior mean lives in the common Cameron-Martin space of the 
two measures. This common Cameron-Martin space is the image space of
$\Sigma^{\frac12}$ in $\h$. Thus we deduce that 
$w:=\prc^{-\frac12}\prc\fp^\ast (\fp \prc\fp^\ast +\nc)^{-1}y\in\h$ almost surely. We next observe that, under $\nu$,
$\Gamma^{-\frac12}y \sim \G(0,SS^*+I)$.  Furthermore
$$w=S^*(SS^*+I)^{-1}\Gamma^{-\frac12}y,$$
thus under $\nu$,
$w\sim N(0, S^\ast(SS^\ast+I)^{-1}S)$ where $S$ is defined in Assumption \ref{Generalinvprobassumption}. Using Lemma \ref{propositionappendix}, 
we thus get that $iv)$ implies that $S^\ast(SS^\ast+I)^{-1}S$ is trace class. Using Lemma \ref{lem:tr}(iv) with $T=(SS^\ast +I)^{-\frac12}S$, 
we then also get that $(SS^\ast +I)^{-\frac12}S S^\ast(SS^\ast +I)^{-\frac12}$ is trace class. Since $(SS^\ast +I)^{\frac12}$ is bounded, using Lemma \ref{lem:tr}(iii) twice we get that $S S^\ast$ is trace class. Finally, again using Lemma \ref{lem:tr}(iv) we get that $S^\ast S$ is trace class, thus ii) holds.
\end{proof}
%\marginpar{\modd{Sergios/Daniel: I get $SS^\ast$ in middle of identity after ``we get that''; thus we get $SS^\ast$ trace-class; do we have enough conditions to know that $S$ is Hilbert-Schmidt so we can comuite?; is it? Also we need that $Q$ is Hilbert-Schmidt?}}
%\mods{I can only see it through the following: $S^\ast(SS^\ast+I)^{-1}S$ is t-c hence $(SS^\ast+I)^{-\frac12}SS^\ast(SS^\ast+I)^{-\frac12}$ is t-c hence by Lemma \ref{lem:tr}(iii) and since $(SS^\ast+I)^\frac12$ is bounded $SS^\ast$ is t-c hence $S^\ast S$ is t-c hence by Lemma \ref{lem:tr}(iii) and since $(S^\ast S+I)^{-1}$ bounded $(S^\ast S+I)^{-1}S^\ast S$ is t-c.}

\subsubsection{Proofs of subsection \ref{ssec:sin2}}   \label{PSS34}

The scalings of $\tau$ and $\effd$ can be readily deduced by comparing the sums defining $\tau$ and $\effd$ with integrals:
\begin{equation*}
\tau(\beta,\gamma,d)\approx \frac{1}{\gamma}\int_1^d \frac{1}{x^\beta} \, dx, \quad \quad \effd \approx \int _1^d \frac{1}{1 +\gamma x^\beta} = \gamma^{-1/\beta} \int_\gamma^{d\gamma^{1/\beta}} \frac{1}{1+y^\beta}\, dy.
\end{equation*} 

Our analysis of the sensitivity of $\rho=\rho(\beta,\gamma,d)$ to the model parameters relies in the following expression for $\rho,$ which is valid unless the effective dimension is infinite, i.e. unless $d=\infty,\, \beta\le 1.$

In the next result, and in the analysis that follows, we ease the notation by using subscripts to denote the coordinate of a vector. Thus we write, for instance, $y_j$ rather than $y(j).$ 
\begin{lemma}\label{lem:nconst}
Under Assumption  \ref{ass1} 
%\begin{align}\label{nconst}\const(y)= \propo\bigl(\rn(\cdot,y)\bigr)=\prod_{j=1}^\infty\sqrt{\frac{\lambda j^\beta}{1+\lambda j^\beta}}\exp\Biggl(\frac{\gamma y_j^2}{2(1+\lambda j^\beta)}\Biggr).
%\end{align}
%Hence, 
\begin{equation}\label{formulaforrho}
\rho=\rho(\beta,\gamma,d):=\prod_{j=1}^d \frac{\frac{j^{-\beta}}{\gamma} + 1}{\sqrt{2\frac{j^{-\beta}}{\gamma} + 1}}  \exp\Biggl(  \sum_{j=1}^d \biggl(   \frac{2}{2+\gamma j^\beta} - \frac{1}{1+\gamma j^\beta}    \biggr)\frac{ y_j^2}{\gamma}  \Biggr),
\end{equation}
\begin{comment}
\begin{equation}\label{formulaforrho}
\rho=\rho(\lambda):=\prod_{j=1}^\infty \frac{1+\lambda j^\beta}{\sqrt{\lambda j^\beta(2+\lambda j^\beta)}}   \exp\Biggl(  \sum_{j=1}^\infty \biggl(   \frac{2}{2+\lambda j^\beta} - \frac{1}{1+\lambda j^\beta}    \biggr)\frac{ y_j^2}{\gamma}  \Biggr),
\end{equation}
\end{comment}
which is finite for $\mar$-almost all $y.$
\end{lemma}
\begin{comment}

\begin{definition}
We define the $d$-th truncation of the the Hilbert space problem described in Assumption \ref{ass1} as any diagonal inverse problem in Euclidean space such that $\fp_d^\ast \fp_d$ and $\prc_d$ have eigenvalues $\{j^{-4\ell}\}_{j=1}^d$ and $\{j^{-2\alpha}\}_{j=1}^d$, respectively. Each of the truncated inverse problems has an associated matrix $A_d(\lambda)$, effective dimensions $\tau_{d}(\lambda)$ and $\effd_{d}(\lambda)$, and an associated parameter $\rho_{d}(\lambda).$
\end{definition}
\end{comment}

\begin{proof}[Proof of Lemma \ref{lem:nconst}]
We rewrite the expectation with respect to $\pi$ as an expectation with respect to the law of $Ku$ as follows. Note that here $u_j$ is a dummy integration variable, which represents the $j$-th corrdinate of  $Ku,$ rather than that of $u.$ Precisely, we have: 
\begin{align*}
\propo\bigl(\rn(\cdot,y)\bigr)&=\int_{\X}\rn(u,y)d\propo(u)\\
&=\int_{\R^{\infty}}\exp\left(-\frac{1}{2\gamma}\sumi u_j^2+\frac{1}{\gamma}\sum_{j=1}^d y_ju_j\right)d\left(\bigotimes_{j=1}^d\G\left(0, j^{-\beta}\right)(u_j)\right)\\ 
&=\prod_{j=1}^d \int_{\R}\exp\left(-\frac{1}{2\gamma}u_j^2+\frac{1}{\gamma} y_ju_j\right)\frac{\exp\left(-\frac{ j^{\beta}u_j^2}2\right)}{\sqrt{2\prm j^{-\beta}}}du_j\\
&=\prod_{j=1}^d \frac1{\sqrt{2\prm  j^{-\beta}}}\int_{\R}\exp\left(-(\gamma^{-1} + j^{\beta})\frac{u_j^2}2+\frac{1}{\gamma} y_ju_j\right)du_j\\
&=\prod_{j=1}^d \frac{\exp\left(\frac{\gamma^{-2}  y_j^2}{2(\gamma^{-1} + j^{\beta})}\right)}{\sqrt{2\prm j^{-\beta}}}\int_{\R}\exp\left(-(\gamma^{-1} + j^{\beta})\frac{\left(u_j-\frac{\gamma^{-1} y_j}{\gamma^{-1} + j^{\beta}}\right)^2}2\right)du_j\\
&=\prod_{j=1}^d\sqrt{\frac{ j^{\beta}}{\gamma^{-1} + j^{\beta}}}\exp\Biggl(\frac{\gamma^{-2}  y_j^2}{2(\gamma^{-1} + j^{\beta})}\Biggr)\\
&=\prod_{j=1}^d\sqrt{\frac{\gamma j^\beta}{1+\gamma j^\beta}}\exp\Biggl(\frac{\gamma^{-1} y_j^2}{2(1+\gamma j^\beta)}\Biggr).
\end{align*}

Thus,
 \[
 \propo\bigl(\rn(\cdot,y)\bigr)^2=\prod_{j=1}^d \frac{\gamma j^\beta}{1+\gamma j^\beta}\exp\Biggl(\frac{\gamma^{-1} y_j^2}{1+\gamma j^\beta}\Biggr)
 \]
and
 \[\propo\bigl(\rn(\cdot,y)^2\bigr)=\prod_{j=1}^d\sqrt{\frac{\gamma j^\beta}{2+\gamma j^\beta}}\exp\Biggl(\frac{2\gamma^{-1} y_j^2}{2+\gamma j^\beta}\Biggr),
\]
Taking the corresponding ratio gives the expression for $\rho.$
\end{proof}

\begin{comment}
-----------------------------------------

The subsequent analysis relies in the heuristic that, since $y_j^2/\gamma$ is order $1,$ we have
$$\rho(\beta,\gamma,d)\approx\prod_{j=1}^d \frac{\frac{j^{-\beta}}{\gamma} + 1}{\sqrt{2\frac{j^{-\beta}}{\gamma} + 1}}   \exp\Biggl(  \effd\bigl(\beta,\frac\gamma2,d\bigr)-\effd(\beta,\gamma,d) \Biggr).$$

If $d\to \infty$ and the smallest eigenvalue $\frac{d^{-\beta}}{\gamma}$ does not converge to 0 as $d\to \infty$, then the product term dominates the exponential one. If the smallest eigenvalue converges to $0$ we have, for large $j$, 
\begin{equation*}
\frac{\frac{j^{-\beta}}{\gamma} + 1}{\sqrt{2\frac{j^{-\beta}}{\gamma} + 1}} \approx 1 + \frac{j^{-2\beta}}{2\gamma^2},
\end{equation*}
and for large $d$
$$\prod_{j=1}^d \frac{\frac{j^{-\beta}}{\gamma} + 1}{\sqrt{2\frac{j^{-\beta}}{\gamma} + 1}}\approx \exp\left(\frac{1}{\gamma^2} \sum_{j=1}^d j^{-2\beta}\right).$$

-----------------------------------------
\end{comment}

\begin{proof}[Analysis of scalings of $\rho$]
Here we show how to obtain the scalings in Table \ref{table:efdtaurhoscalings}. 
Taking logarithms in \eqref{formulaforrho}
\begin{equation}\label{twosums}
\log(\rho)=\sum_{j=1}^d \log\Biggl(  \frac{\frac{j^{-\beta}}{\gamma} + 1}{\sqrt{2\frac{j^{-\beta}}{\gamma} + 1}} \Biggr)    +    \sum_{j=1}^d \biggl(   \frac{2}{2+\gamma j^\beta} - \frac{1}{1+\gamma j^\beta}    \biggr)\gamma^{-1} y_j^2. 
\end{equation}
Note that every term of both sums is positive. In the small noise regimes the first sum dominates, whereas in the large $d$, $\beta \searrow 1$ the second does.
We show here how to find the scaling of $\gamma \to 0$ when $d=\infty.$ 

We have that 
\begin{align*}
\log(\rho)&\ge \sum_{j=1}^\infty \log\Biggl(  \frac{\frac{j^{-\beta}}{\gamma} + 1}{\sqrt{2\frac{j^{-\beta}}{\gamma} + 1}} \Biggr) \\
&\approx \int_1^{f(\gamma)} \log\Biggl(  \frac{\frac{x^{-\beta}}{\gamma} + 1}{\sqrt{2\frac{x^{-\beta}}{\gamma} + 1}} \Biggr) \,dx + \int_{f(\gamma)}^\infty \log\Biggl(  \frac{\frac{x^{-\beta}}{\gamma} + 1}{\sqrt{2\frac{x^{-\beta}}{\gamma} + 1}} \Biggr) \,dx
\end{align*}
where $f(\gamma)$ is a function of $\gamma$ that we are free to choose. Choosing $f(\gamma)= \gamma^{-1/\beta - \epsilon}$ ($\epsilon$ small) the first integral dominates the second one and, for small $\gamma,$
$\log(\rho) \ge \gamma^{-1/\beta - \epsilon} \log(\gamma^{-\epsilon \beta /2})$
from where the result in Table \ref{table:efdtaurhoscalings} follows. The joint large $d$, small $\gamma$ scalings can be established similarly.

When the second sum in \eqref{twosums} dominates, the scalings hold in probability. To illustrate this, we study here how to derive the large $d$ limit with $\beta<1$. 
Without loss of generality we can assume in what follows that each $y_j$ is centered, i.e. $y_j\sim \G(0,\gamma)$ instead of $y_j\sim \G\bigl((Ku)_j^\dagger,\gamma\bigr).$ This is justified since, for any $c>0,$
\begin{equation*}
\rp(y_j^2 \ge c) = \rp( |y_j|\ge c^{1/2}) \ge \rp( |y_j- (Ku)_j^\dagger|\ge c^{1/2}).
\end{equation*}
Neglecting the first sum in \eqref{twosums}, which can be shown to be of lower order in $d$, we get
\begin{equation*}
\sum_{j=1}^d \biggl(   \frac{2}{2+\gamma j^\beta} - \frac{1}{1+\gamma j^\beta}    \biggr)\gamma^{-1} y_j^2 = S(y,d). 
\end{equation*}
Using that $\E y_j^2 = \gamma$,
\begin{align*}
\E \log(\rho)&\ge  \sum_{j=1}^d \biggl(   \frac{2}{2+\gamma j^\beta} - \frac{1}{1+\gamma j^\beta}    \biggr)\\
&\approx  \int_{1}^d \biggl(   \frac{2}{2+\gamma x^\beta} - \frac{1}{1+\gamma x^\beta}    \biggr) \, dx \approx d^{1-\beta}=:m(d).
\end{align*}
Also, since $\Var(y_j^2) = 3\gamma^{2},$ 
\begin{align*}
\Var \log(\rho)&\ge  \sum_{j=1}^d \biggl(   \frac{2}{2+\gamma j^\beta} - \frac{1}{1+\gamma j^\beta}    \biggr)^2 \gamma^2\\
&\approx  \int_{1}^d \biggl(   \frac{2}{2+\gamma x^\beta} - \frac{1}{1+\gamma x^\beta}    \biggr)^2 \, dx \approx d^{1-2\beta}=:c(d).
\end{align*}

Thus we have 
\begin{align*}
\rp\Big(\log(\rho) \ge m(d)/2\Big) &\ge \rp\Big(S(y,d) \ge m(d)/2\Big)\\
&\ge \rp\Big(S(y,d) \ge \E S(y,d)/2\Big)\\
&\ge \rp\Big(|S(y,d)- \E S(y,d)| \le \E S(y,d)/2\Big)\\
&= 1- \rp\Big(|S(y,d)- \E S(y,d)| \ge \E S(y,d)/2\Big)\\
&\ge 1- \rp\Big(|S(y,d)- \E S(y,d)| \ge m(d)/2\Big)\\
&\ge 1- 4\frac{c(d)}{m(d)^2} \to 1.
\end{align*}
\end{proof}

\subsection{Proofs Section \ref{sec:FIL}}   \label{PS4}
\label{sec:PFIL}
The following lemma will be used in the proof of Theorem \ref{theoremcomparisoncollapse}. It justifies the use of the cyclic property in calculating certain traces in the infinite dimensional setting.

\begin{lemma}
Suppose that $A=S^*S,$ where $S=\Gamma^{-1/2}K\Sigma^{1/2}$ as in Assumption \ref{Generalinvprobassumption} is bounded. Then 
\begin{equation*}
\tau = \tr(A) = \tr( \Gamma^{-1} K \Sigma K^*).
\end{equation*}
 Therefore, using the equivalence in Table \ref{table:filteringequiv} we have that $\tau_{st}$ and $\tau_{op}$  admit the following equivalent expressions:
\begin{equation}
\label{eq:t1}
\tau_{st}=\tr\bigl(R^{-1} H (MPM^* + Q) H^* \bigr)
\end{equation}
and
\begin{equation}
\label{eq:t2}
\tau_{op}=\tr\bigl((R+HQH^*)^{-1} HMPM^*H^*\bigr).
\end{equation}
\end{lemma}
\begin{proof}
 Using Lemma \ref{lem:tr}(iv) we have that 
$\tau = \tr(S^*S) = \tr(S S^*).$ Now note that $SS^*=\Gamma^{-1/2} K \Sigma K^* \Gamma^{-1/2}$ is bounded since $A$ is, and that $\Gamma^{1/2}$ is also bounded, hence we can use Lemma \ref{lem:tr}(v) to get the desired result.
\end{proof}

\begin{proof}[Proof of Theorem \ref{theoremcomparisoncollapse}]
Using the previous lemma,
\begin{align*}
\tau_{st} &= \tr\Bigl(R^{-1} HMPM^*H^* \Bigr) + \tr\Bigl(R^{-1}  HQH^* \Bigr) \\
&\ge \tr\Bigl( R^{-1}  HMPM^*H^* \Bigr)\\
&\ge \tr\Bigl((R+HQH^*)^{-1}  HMPM^*H^* \Bigr)=\tau_{op},
\end{align*}
where the first inequality holds because $R$ is positive-definite and $HQH^*$ is positive semi definite, and the second one follows from Lemma \ref{lem:tr}(vi).

If $\tr(HQH^* R^{-1})<\infty$ then there is $c>0$ such that, for all $x,$ $\|HQH^*x\| \le c \|Rx\|$. Hence applying again Lemma \ref{lem:tr}(vi) for both directions of the equivalence, we obtain that
\begin{align*}
\tau_{op}=\tr\Bigl((R+HQH^*)^{-1} HMPM^*H^* \Bigr)<\infty &\iff \tr\Bigl( R^{-1} HMPM^*H^* \Bigr)<\infty \\
&\iff \tau_{st}<\infty.
\end{align*}
\end{proof}

\end{document}